\documentclass[nofootinbib,onecolumn,pra,aps]{revtex4}

\usepackage[T1]{fontenc}
\fontencoding{T1}  
\usepackage[utf8]{inputenc}

\usepackage{appendix}
\usepackage{tikz}
\usepackage{amsfonts}
\usepackage{amsmath,amssymb,amsthm}
\usepackage{graphicx}
\usepackage{fancyhdr}
\usepackage{tcolorbox}
\usepackage{mathtools}
\usepackage{array}
\usepackage{bbm}
\usepackage[T1]{fontenc}
\usepackage{url}
\usepackage{placeins}
\usepackage{enumerate}
\usepackage{braket}
\usepackage{dsfont}
\usepackage{booktabs}
\usepackage[plain]{fancyref} 
\usepackage{algorithm}
\usepackage{algpseudocode}
\usepackage{subfig}
\usepackage{bbm}
\usepackage{hyperref}

\theoremstyle{plain}

\newtheorem{example}{Example}[section]
\theoremstyle{remark}

\newcommand*{\fancyrefthmlabelprefix}{thm}
\newcommand*{\fancyreflemlabelprefix}{lem}
\newcommand*{\fancyrefcorlabelprefix}{cor}
\newcommand*{\fancyrefdefilabelprefix}{defi}
\frefformat{plain}{\fancyreflemlabelprefix}{lemma\fancyrefdefaultspacing#1}
\Frefformat{plain}{\fancyreflemlabelprefix}{Lemma\fancyrefdefaultspacing#1}
\frefformat{plain}{\fancyrefthmlabelprefix}{theorem\fancyrefdefaultspacing#1}
\Frefformat{plain}{\fancyrefthmlabelprefix}{Theorem\fancyrefdefaultspacing#1}
\frefformat{plain}{\fancyrefcorlabelprefix}{corollary\fancyrefdefaultspacing#1}
\Frefformat{plain}{\fancyrefcorlabelprefix}{Corollary\fancyrefdefaultspacing#1}
\frefformat{plain}{\fancyrefdefilabelprefix}{definition\fancyrefdefaultspacing#1}
\Frefformat{plain}{\fancyrefdefilabelprefix}{Definition\fancyrefdefaultspacing#1}
\newcommand*{\fancyrefalglabelprefix}{alg}
\newcommand*{\frefalgname}{algorithm}
\newcommand*{\Frefalgname}{Algorithm}
\frefformat{plain}{\fancyrefalglabelprefix}{%
	\frefalgname\fancyrefdefaultspacing#1%
}%
\Frefformat{plain}{\fancyrefalglabelprefix}{%
	\Frefalgname\fancyrefdefaultspacing#1%
}%

\def\beq{\begin{equation}}
\def\eeq{\end{equation}}
\def\bq{\begin{quote}}
	\def\eq{\end{quote}}
\def\ben{\begin{enumerate}}
	\def\een{\end{enumerate}}
\def\bit{\begin{itemize}}
	\def\eit{\end{itemize}}

\def\r|{\right|}

\newcommand\M{\mathbb{M}}

\newcommand{\cC}{{\mathbb{C}}}
\newcommand{\mat}{{\rm{mat}}}

\newcommand{\id}{{\rm{id}}}

\newcommand{\ot}{\otimes}

\newcommand{\End}{\operatorname{End}}
\newcommand{\Ad}{\operatorname{Ad}}
\newcommand{\Int}{\operatorname{Int}}
\newcommand{\sgn}{\operatorname{sgn}}

\newcommand{\<}{\langle}

\renewcommand{\>}{\rangle}
\newcommand{\tr}{\operatorname{tr}}
\newcommand\be{\begin{equation}}
\newcommand\ee{\end{equation}}

\usepackage{hyperref}
\hypersetup{colorlinks, citecolor=magenta, filecolor=blue, linkcolor=blue, urlcolor=green}
\newtheorem{theorem}{Theorem}

\newtheorem{corollary}[theorem]{Corollary}
\newtheorem{fact}[theorem]{Fact}

\newtheorem{definition}[theorem]{Definition}
\newtheorem{p_example}[theorem]{Example}  

\newtheorem{lemma}[theorem]{Lemma}

\newtheorem{proposition}[theorem]{Proposition}

\begin{document}	
	\title{Positive Maps From Irreducibly Covariant Operators}
	
\author{Piotr Kopszak$^{1}$,
	Marek Mozrzymas$^{1}$ and Micha{\l} Studzi\'nski$^{2}$}
\affiliation{$^1$ Institute for Theoretical Physics, University of Wrocław
	50-204 Wrocław, Poland \\
	$^2$ Institute of Theoretical Physics and Astrophysics, National Quantum Information Centre, 
Faculty of Mathematics, Physics and Informatics, University of Gda{\'n}sk, Wita Stwosza 63, 80-308 Gda{\'n}sk, Poland}
\date{\today}	 
\begin{abstract}
In this paper, we discuss positive maps induced by (irreducibly) covariant linear operators for finite groups. The application of group theory methods allows deriving some new results of a different kind. In particular, a family of necessary conditions for positivity, for such objects is derived, stemming either from the definition of a positive map or the novel method inspired by the inverse reduction map.  In the low-dimensional cases, for the permutation group $S(3)$ and the quaternion group $Q$, the necessary and sufficient conditions are given, together with the discussion on their decomposability. In higher dimensions, we present positive maps induced by a three-dimensional irreducible representation of the permutation group $S(4)$ and $d$-dimensional representation of the monomial unitary group $MU(d)$. In the latter case, we deliver if and only if condition for the positivity and compare the results with the method inspired by the inverse reduction. We show that the generalised Choi map can be obtained by considered covariant maps induced by the monomial unitary group. As an additional result,  a novel interpretation of the Fujiwara-Algolet conditions for positivity and complete positivity is presented. Finally, in the end, a new form of an irreducible representation of the symmetric group $S(n)$ is constructed, allowing us to simplify the form of certain Choi-Jamio{\l}kowski images derived for irreducible representations of the symmetric group.   
\end{abstract}
\maketitle
\section{Introduction}
Symmetry is one of the most prominent properties which physical theory can exhibit. Usually, it is induced by some group structure. In this paper we consider symmetries encoded in linear operators of the special importance in quantum information theory that are \emph{irreducibly covariant} with respect to some finite group $G$. We call a map $\Phi$  
irreducibly covariant ($ICLM$-irreducibly covariant linear map) if  
\begin{equation}
\label{1a}
 \Phi\left(U(g)\rho U^\dagger(g)\right) = U(g)\Phi(\rho)U^\dagger(g),\quad \forall g \in G,
\end{equation}
for some unitary irreducible representation $U$ of a finite group $G$. If $U$ is not irreducible representation we say that the map is just covariant.

There are two special classes of linear maps that play a special role in quantum information theory: positive maps $(P)$ and completely positive $(CP)$ maps. Namely, a map $\Phi$, not necessarily covariant, is \textit{positive} if for every positive semidefinite matrix $M\geq0$, we have $  \Phi(M)\geq0$ as well. Secondly, a map $\Phi$ is \emph{completely positive}  when $\Phi \otimes \text{id}$, where $\text{id}$ stands for the identity map, is positive. Completely positive maps model quantum channels, that are among the fundamental blocks for quantum information processing tasks~\cite{rev1}. On the other hand, maps that are positive, but not completely positive, which are of the main interest of this manuscript, are useful since they can be used for entanglement detection via so-called entanglement witnesses. The concept of entanglement witness follows from the famous Hanh-Banach Theorem and was introduced firstly in the field of quantum information by Terhal in~\cite{ter}.
This method allows us to detect entanglement without the full knowledge about a given quantum state since one does not have to apply the full tomography of the state, which is an expensive process. Instead of that, it is enough to check the mean value of the observable representing entanglement witness on a given quantum state.  Such a method is universal since any entangled state posses corresponding entanglement witness. 

In general, the problem of classification or even finding some new examples of positive maps with certain properties is very hard despite many attempts and important results in the field~\cite{Woronowicz,Chrust1,Takasaki1983,TOMIYAMA1985169,Chru2009,Kossakowski2003,Chru2014,Chru2009a,TANG198633}. The main reason for that is the non-existence of a universal operational criterion for positivity. To prove that a given map is positive one has to check for example its block positivity, i.e. positivity of expectation values on all product vectors, while to check complete positivity of the given map it is enough to compute only all eigenvalues of the corresponding Choi-Jamio{\l}kowski image, which can be done effectively, at least for reasonably low dimensions. Imposing symmetry constraints simplifies the problem significantly and thus can lead to new results concerning the positivity of linear maps in the considered class and produce new important examples. One such constraint is irreducible covariance with respect to a given finite group $G$, defined as in~\eqref{1a}. The $ICLM$ with respect to a group $G$ inherits some symmetrical patterns from the group $G$, which usually simplifies the conditions for the positivity and complete positivity.

Approach exploiting covariance so far has been applied with success for the unitary group $\mathcal{U}(d)$, i.e. considered map from~\eqref{1a} commutes with every unitary matrix, not necessarily an irrep. These studies has been started by Bhat in~\cite{bhat2011}, where it was proven that maps satisfying~\eqref{1a} for $U=V\in \mathcal{U}(d)$ must be of the form $\Phi(X)=\alpha X+\beta \tr(X)\mathbf{1}$, where $\alpha,\beta \in\mathbb{C}$. It is clear that properties such as $P$ or $CP$ of $\Phi$ are fully determined by the values of the parameters $\alpha,\beta$. Later, in~\cite{COLLINS2018398} authors have generalized this idea by considering a family of three-parameter covariant maps $\Phi(\alpha,\beta,n):\mathbb{M}(n,\mathbb{C})\rightarrow \mathbb{M}(n^2,\mathbb{C})$, where $\alpha,\beta\in\mathbb{C}$, $n\in\mathbb{N}$. One of the main result is the link between $k-$positivity (in particular $P$) of a linear map and its covariance property. In particular, authors derive range of parameters $\alpha,\beta$ with $n=3$, for which considered map $\Phi(\alpha,\beta,3)$ is $2-$positive, $CP$ and decomposable. However, the if and only if conditions for $P$ are presented for an arbitrary $n\geq 3$. Next, in~\cite{bardet2018characterization} further generalisations are discussed. Namely, authors using elements of the representation theory and graphical representation of permutations, have characterised linear maps $\Phi:\mathbb{M}(n,\mathbb{C})\rightarrow \mathbb{M}(n^a,\mathbb{C})\ot \mathbb{M}(n^b,\mathbb{C})$ satisfying  $\Phi(UXU^{\dagger})=(\overline{U}^{\ot a}\ot U^{\ot b})\Phi(X)(\overline{U}^{\ot a}\ot U^{\ot b})^{\dagger}$, for all $X\in \mathbb{M}(n,\mathbb{C})$ and $a,b \in\mathbb{N}$. Recently, in paper~\cite{huber2020positive} such covariance, with $a=0$, plays the central role in extension of the polarized Cayley-Hamilton identity to an operator inequality on the positive cone and characterisation of the set of multilinear covariant positive maps. Finally, we mention 
9
 the result from~\cite{Siudzi_ska_2018}, where the family of positive maps induced by the covariance with respect to the finite group generated by the Weyl operators is derived. 
Our paper is a complementary continuation of~\cite{Moz2017}, where the authors analysed completely positive irreducibly covariant linear maps. The difference with the results quoted above is that now a map $\Phi$ is covariant with respect to some set of unitary matrices, which are the irreps of a chosen finite group $G$, so our requirement is less strong, more demanding in the analysis, causing the impossibility of the use results derived in papers cited above. Nevertheless, we still are able to derive general results.

Our paper is organised as follows. We start from Section~\ref{Def_notations} with the general introduction of the theory of entanglement. In particular, we define the Choi-Jamio{\l}kowski image, its connection with $P$ and $CP$ as well as with entanglement witnesses.
Further, in Section~\ref{general_cons} we present general considerations about the $ICLMs$ of the form
\begin{equation}
\label{dede}
    \Phi=\sum_\alpha l_\alpha\Pi^\alpha,\qquad l_\alpha \in \mathbb{C},
\end{equation}
where $\alpha$ labels irreps of $G$, $\Pi^\alpha$ are projectors arising from the structure of group $G$ and its unitary representation $U$ given in \eqref{1a}, and $l_\alpha$ are spectral parameters of such decomposition. Further, by $\operatorname{id}$ we understand the identity representation. Since the projectors $\Pi^\alpha$ can be calculated for a given representation the problem of checking whether a given map $\Phi$ is $P$ is reduced to determining the spectral parameters $l_\alpha$. The considerations are of the general nature, since we do not choose a specific representation $U$, but rather use its general properties.

 In the mentioned paper~\cite{Moz2017}, the necessary and sufficient conditions for $CP$ has been presented as a restriction on the values of the spectral parameters $l_{\alpha}$. We were able to provide a condition of similar structure, describing maps that are positive, but not completely positive. However, unlike in the former case, this condition is not operational. Nevertheless, for such a large class of $ICLMs$ given in~\eqref{dede}, we prove a general necessary condition for their positivity in the form of a condition on their spectral parameters $l_{\alpha}$. It has a simple geometrical meaning: in the space of spectral parameters of $\Phi$, all the spectral parameters $l_{\alpha}$, such that $\alpha\neq \operatorname{id}$, must lie inside a cuboid, whose edge lengths are equal to $2l_{\operatorname{id}}$. It appears that this necessary condition is non-trivial, because we prove additionally, that for the $ICLMs$ generated by two-dimensional irreps of the group $S(3)$ and the quaternion group $Q$, the necessary condition is also the sufficient one.  
 Additionally, we deliver the explicit decomposition into the sum of complete positive and complete co-positive maps of the $ICLM's$ induced by the mentioned groups.  It is known, due to~\cite{Woronowicz}, that for qubits, such a decomposition always exists, but without any method of finding it in a given case. 
 
  To give a flavour of our considerations, in Figure \ref{fig:S3_intro} we present two regions for which linear map $\Phi$, which is irreducibly covariant with respect to symmetric group $S(3)$, admitting the decomposition from~\eqref{dede}, is positive but not completely positive. One is calculated directly from the definition and describes the necessary and sufficient region, while the second one is obtained from the novel method based on the inverse reduction map (Section~\ref{invRR}), which is more universal but provides only subregion of positivity. 

\begin{figure}[h]
     \centering
     \subfloat[][]{\includegraphics[width=0.3\textwidth]{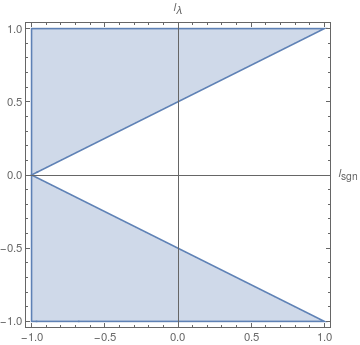}\label{<figure1>}}
     \subfloat[][]{\includegraphics[width=0.3\textwidth]{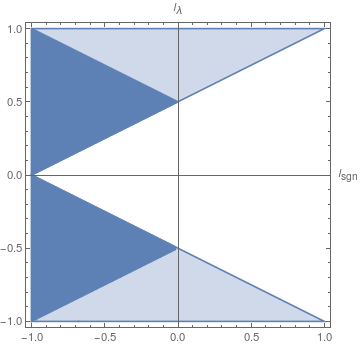}\label{<figure2>}}
    \caption{Possible solutions for the spectral parameters of $ICLM$ $\Phi$ (see~\eqref{dede}) with respect to $S(3)$ group, and its two-dimensional irrep labelled by $\lambda$. The region depends on two parameters $l_\text{sgn}, l_\lambda$, where $\sgn$ is the sign representation. The panel (a) shows the general solution where the map $\Phi$ is positive, but not completely positive, whereas in the panel (b) the dark blue region corresponds to the result obtained from the inverse reduction map.}
    \label{fig:S3_intro}
\end{figure}
\FloatBarrier
In Section~\ref{invRR} we deliver a novel method for the construction of positive $ICLMs$  for irreducible representations of dimensions higher than the qubit case. The method is based on the inverse reduction map~\cite{Moz2015} and it rids of the explicit checking the block positivity of the corresponding Choi-Jamio{\l}kowski image. We reduce the problem of checking the block positivity to a set of linear inequalities describing the region of positivity, which can be easily solved. Although we consider $ICLMs$, this method can be applied to any other linear map which proves the universality of the result. In particular, it can be applied to symmetric groups of higher order than $S(3)$ discussed in Section~\ref{general_cons}.
In Section~\ref{large_d}, we deliver examples of positive maps acting in higher dimensions. We examine a positive $ICLMs$ for the three-dimensional irreducible representation of $S(4)$ and the $d-$dimensional representations of the monomial unitary group. Here, we use two approaches: the inverse reduction map and the direct examination of the block-positivity of the Choi-Jamio{\l}kowski image of a given $ICLM$. 
While the latter method is of a little use in the general scenario, it proved efficient in the special case of determining positivity of $ICLMs$ with respect to the subgroup of monomial unitary group $MU(d)$ denoted as $MU(d,n)$, for the parameter $n\in \mathbb{N}$, and an arbitrary dimension $d$.  This map has been defined in~\cite{Daniel} and used in the same paper to generalize randomized benchmarking protocol, when gates are representations of a  finite group, but not necessarily irreducible or 2-design. Moreover, it turns out the monomial group $MU(d)$ plays a role in many-body state formalism and some aspects of quantum computations, for the details please see~\cite{Nest} and references within it, while its subgroup is also in the interest of quantum information community. It contains so-called $T-$gate~\cite{NielsenChuang}, defined as  $T=|0\>\<0|+\exp(\frac{\operatorname{i}\pi}{4})|1\>\<1|$, whenever $n\geq 8$. This gate, together with the Clifford gates, forms a universal set for quantum computations~\cite{NielsenChuang}.  The result obtained by the direct approach gives the necessary and sufficient conditions for the parameters of the  $ICLM$ in question, which is a significant improvement over conditions obtained from the inverse reduction map approach. The comparison of these two methods is also presented. In Section~\ref{Choi-porownanie} we compare results obtained in the previous Section, concerning qutrit ($d=3$) case, with the famous generalized Choi map~\cite{cho}, showing that our maps include the generalized Choi map and generate a larger region of positivity. Unfortunately, we could not find any new indecomposable positive maps in those regions. Nevertheless, for at least some cases of decomposable maps, we provided an explicit form of the decomposition. Finally, in Section~\ref{sec:unit_qc} yet another remarkable property of irreducible covariance is presented. It turns out, that in the qubit case every unital quantum channel can be expressed in the same terms as a quantum channel which is irreducibly covariant with respect to the two-dimensional irrep of the quaternion group $Q$. We show that Fujiwara-Algolet conditions~\cite{PhysRevA.59.3290} for $P$ and $CP$ for arbitrary qubit unital map obtained in this approach exploiting Bloch sphere representation coincide with the conditions for $P$ and $CP$ for quantum channels that are irreducibly covariant with respect to quaternion group. Thus, the conditions for $P$ and $CP$ for an arbitrary qubit unital map are connected with the quaternionic invariance of quantum channels.
\section{Definitions and notations}
\label{Def_notations}
In this section, we provide the necessary mathematical tools for our further considerations. In the first part, we present basic notions from linear algebra and entanglement theory. We focus on the concepts of an entanglement witness together with their (in)decomposability and their connection with positive maps. The second part is dedicated to the connection of the linear maps and their Choi-Jamio{\l}kowski images. We focus on the \textit{if and only if} conditions for positivity, complete positivity and complete copositivity for an arbitrary linear map. In particular, we introduce the group-theoretic background for studying linear maps which are irreducibly covariant.

\subsection{Quantum entanglement, entanglement witnesses and decomposability}
\label{wit_thm}
Let $\M (d, \cC)$ denote the space of $d\times d$ complex matrices and let
     $\{E_{ij}\}_{i,j=1}^d$, where $E_{ij} \equiv \ket{i}\bra{j}$, denote a basis
     of $\M (d, \cC)$ and by $\{|i\>\}_{i=1}^d$ we denote the standard basis for $\mathbb{C}^d$. The matrix resulting from the action of  map $\Phi \in \End[\M(d,\mathbb{C})]$ on any basis element $E_{ij} \in \M(d, \cC)$ can be expressed as $\Phi (E_{ij})=\sum_{k,l=1}^d\phi _{kl,ij}E_{kl}$.
The coefficients $\phi _{kl,ij}$ can be viewed as elements of a $d^2 \times d^2$ matrix, and hence we use the notation:
\begin{align}
\mat(\Phi) :=\left(\phi _{kl,ij}\right) \in \M(d^2, \cC).
\end{align}
In this and further sections by $\mathcal{H}=\mathbb{C}^d$ we denote Hilbert space of dimension $d$. Using this notation we define a set
\be
\mathcal{S}(\mathcal{H}):=\{\rho \in \mathcal{H} \ | \ \rho\geq 0, \ \tr \rho=1\}
\ee
of all states on $\mathcal{H}$. Having two Hilbert spaces $\mathcal{H}, \mathcal{K}$ we say that state $\rho\in \mathcal{S}(\mathcal{H}\ot \mathcal{K})$ is \textit{separable} if it can be written as $\rho=\sum_i p_i \sigma_i \ot \omega_i$, where $\sigma_i,\omega_i$ are states on $\mathcal{H},\mathcal{K}$ respectively, and $p_i$ are positive numbers satisfying $\sum_i p_i=1$, otherwise the state $\rho$ is \textit{entangled}.
Having above we are in a position to define objects called \textit{entanglement witnesses}~\cite{ter}:
\begin{definition}
	\label{def01}
	The hermitian operator $W \in \M(d^2,\mathbb{C})$ is called entanglement witness when:
	\begin{enumerate}
		\item $W \ngeq  0$,
		\item $\tr \left(\sigma W\right)\geq 0$ for all separable states $\sigma$,
		\item There exists at least one entangled state $\rho$, such that $\tr \left(\rho W\right)< 0$.
	\end{enumerate}
\end{definition}
	 The so-called \textit{Choi-Jamiołkowski isomorphism} gives useful characterization of entanglement witnesses. For a linear map $\Phi \in \End[\M(d,\mathbb{C})]$ its Choi-Jamio{\l}kowski image $J(\Phi)$ is given by~\cite{Jam,cho}:
	\be
	\label{Ch-J}
	J(\Phi):=\sum_{i,j=1}^dE_{ij}\ot \Phi(E_{ij})\in \M(d^2,\mathbb{C}).
	\ee
	Isomorphism defined in~\eqref{Ch-J} encodes properties of linear maps into properties of corresponding Choi-Jamio{\l}kowski image. 

Thanks to Choi-Jamio{\l}kowski isomorphism every entanglement witness is connected with positive, but not completely positive linear map  $\Phi \in \End[\M(d,\mathbb{C})]$ by
\be
\label{connP}
W=(\mathbf{1}\ot \Phi)P_d^+,
\ee
where $P_d^+$ is a projector on \textit{maximally entangled state} $|\psi^+_d\>=\frac{1}{\sqrt{d}}\sum_{i=1}^d|ii\>$. In the set of all entanglement witnesses we distinguish subset of \textit{decomposable} ones. Namely any decomposable entanglement witness $W$ admits the following decomposition
\be
\label{decompos}
W=A+B^{\Gamma},
\ee
with $A,B$ being positive operators on the space $\mathcal{H}\ot \mathcal{H}$, and $\Gamma$ being a partial transposition with respect to the standard basis. The corresponding map $\Phi_W$ is decomposable if it can be expressed as 
\begin{equation}
\Phi_W = \Phi^{(A)} + \Phi^{(B)}\circ T,
\end{equation}
where $\Phi^{(A)}, \Phi^{(B)}$ are completely positive, and $T$ stands for transposition with respect to standrard basis. Additionally, operators $W$ satisfying Definition~\ref{def01} for states being \textit{PPT entangled}, i.e. states $\rho$ for which $\rho^{\Gamma}\geq 0$, do not admit decomposition~\eqref{decompos} and they are called \textit{indecomposable entanglement witnesses}. At this point for more informations about entanglement witnesses and their properties we refer reader to a review paper~\cite{Chrust1}.

\subsection{Irreducibly covariant linear maps and quantum channels}
\label{S_ICLM}
Here we summarize all necessary facts and definitions about irreducibly covariant linear maps and we explain the connection of the Choi-Jamio{\l}kowski image of a general linear map with its positivity, complete positivity as well as complete co-positivity. For more details we refer reader to~\cite{Moz2017}, we also keep the original notation from the mentioned paper, since we consider here the complementary topic. Additionally, we prove to auxiliary facts, necessary for the further studies on the positivity.
\begin{definition}
\label{iclm}
A linear map $\Phi \in \End[\M(d,\mathbb{C})]$ is irreducibly covariant with respect to an irreducible representation $U$ of the finite group $G$  if
\be
\label{irr_cov}
\forall g\in G \quad \forall X\in \M(d,\mathbb{C}) \quad \Phi\left(U(g)X U^{\dagger}(g) \right)=U(g)\Phi(X)U^{\dagger}(g).
\ee
\end{definition}
The above definition can be phrased in the equivalent way, using concept of the adjoint map $
\Ad_{U}^{G}:G\longrightarrow \End\left[ \M(n,\mathbb{C})\right] $ defined through its action 
on any $X \in \M(n , \cC)$ as follows:
\begin{align}
\forall g \in G \quad \Ad_{U(g)}(X):= U(g)XU^{\dagger }(g).
\end{align}
Having that, we say that a linear map $\Phi\in \End[\M(d,\mathbb{C})]$ is irreducibly covariant if $\forall g\in G,\quad \forall X\in \M(n,\mathbb{C})\quad \Ad_{U(g)}[\Phi
(X)]=\Phi \lbrack \Ad_{U(g)}(X)]$. For further purposes, we define also  the commutant of the adjoint representation $\Int_G(\Ad_U)$ and its matrix representation $\Int_G(U\ot \overline{U})$ given as
\be
\Int_G(\Ad_U):=\left\lbrace \Psi\in \End[\M(d,\mathbb{C})] \ : \ \Psi\circ \Ad_U=\Ad_U\circ \Psi \right\rbrace,
\ee
and 
\be
\Int_G(\operatorname{mat}(\Ad_{U}^{G}))=\Int_G(U\ot \overline{U}):=\left\lbrace X\in \M(d^2,\mathbb{C}) \ : \ \forall g\in G \quad X\left(U(g)\ot \overline{U}(g) \right)=\left(U(g)\ot \overline{U}(g) \right)X  \right\rbrace.
\ee
Obviously a linear map $\Phi \in \End[\M(d,\mathbb{C})]$ is covariant with respect to $G$ and its irreducible representation $U$ when belongs to $\Int_G(\Ad_U)$ or equivalently $\operatorname{mat}(\Phi)$ belongs to $\Int_G(U\ot \overline{U})$.
In the most of the cases we study properties of $\Int_G(U\ot \overline{U})$ in the matrix space $\M(d^2,\mathbb{C})$, which is simpler (but equivalent) than studying $\Int_G(\Ad_U)$ in the space $\End[\M(d,\mathbb{C})]$. As was shown in~\cite{FHa} the representation $U\ot U^c:G\rightarrow \M(d^2,\mathbb{C})$ is not irreducible and we have
\be
\label{UUmult}
U\ot \overline{U}=\bigoplus_{\alpha \in \Theta} m_{\alpha}\varphi^{\alpha},
\ee
where $\Theta$ is the set of irreps of the group $G$ that appear in the decomposition $U\ot \overline{U}$, $\varphi ^{\alpha}$ are unitary irreps of the group $G$: $\forall \, g \in G$, $%
\varphi ^{\alpha}(g)=\left( \varphi _{ij}^{\alpha }(g)\right) \in \M\left(
d_{\alpha} ,\mathbb{C}\right) $,  and $m_{\alpha }$  is the multiplicity of the irrep $\varphi^{\alpha }$ of dimension $d_{\alpha}\equiv \dim \varphi^{\alpha }$. The multiplicity $m_{\alpha}$ is given by the 
following expression~\cite{NS}:
\begin{equation}
\label{mult}
m_{\alpha }=\frac{1}{|G|}\sum_{g\in G}\chi ^{\alpha }\left( g^{-1}\right) \chi ^{\Ad}(g),
\end{equation}
 $\chi ^{\alpha }$ is the character of the irrep $\varphi ^{\alpha }(g) $,  and $\chi^{\Ad_U}(g)=\left|\chi^U(g) \right|^2$ is the   character of the
adjoint representation $\Ad_{U}$. The identity irrep $\varphi^{\id}$ is always included in the decomposition~\eqref{UUmult} with the multiplicity one. 

If we restrict ourselves to the case when decomposition~\eqref{UUmult} is multiplicity free, i.e. all irreps in the set $\Theta$ occur with the multiplicity at most one.  In this situation we can write
\be
\Int_G(\Ad_U)=\operatorname{span}_{\mathbb{C}}\left\lbrace \Pi^{\alpha} \ : \ \alpha\in \Theta \right\rbrace.
\ee
The operators $\Pi^{\alpha}\in \End[\M(d,\mathbb{C})]$ are of the form
\be
\label{op_Pi}
\Pi^{\alpha}=\frac{d_{\alpha}}{|G|}\sum_{g\in G}\chi^{\alpha}\left(g^{-1} \right)\Ad_{U(g)}, \quad \alpha \in \Theta,
\ee
and
\be
\label{orto_Pi}
\Pi^{\alpha}\Pi^{\beta}=\delta_{\alpha \beta}\Pi^{\alpha},\quad (\Pi^{\alpha})^*=\Pi^{\alpha},\quad \sum_{\alpha \in \Theta}\Pi^{\alpha}=\id_{\End[\M(d,\mathbb{C})]},
\ee
where $\id_{\End[\M(d,\mathbb{C})]}$ is identity map on $\End[\M(d,\mathbb{C})]$. Equivalently we can write
\be
\Int_G(U\ot U^c)=\operatorname{span}_{\mathbb{C}}\left\lbrace \widetilde{\Pi}^{\alpha} \ : \ \alpha \in \Theta \right\rbrace, 
\ee 
where 
\be
\label{mat_Pi}
\widetilde{\Pi}^{\alpha}=\mat(\Pi^{\alpha})=\frac{d_{\alpha}}{|G|}\sum_{g\in G}\chi^{\alpha}\left(g^{-1} \right)U(g)\ot \overline{U}(g).
\ee
Orthogonality properties from~\eqref{orto_Pi} are also valid for matrix images $\widetilde{\Pi}^{\alpha}$. Since, the projectors from~\eqref{ICQC} have group-theoretical origin, they admit useful properties. Namely, we have the following:
\begin{fact}
	\label{F3}
	Let $\Pi^\alpha$ be projectors given through expression \eqref{ICQC}. Then we have, for all $X\in \M(|U|, \mathbb{C})$:
	\begin{itemize}
		\item $\Pi^{\operatorname{id}}(X)=\frac{1}{|U|}\text{Tr}(X)\mathbf{1}_{|U|},$
		\item $\alpha\neq {\operatorname{id}} \implies\text{Tr}\Pi^\alpha(X)=0. $
	\end{itemize}
The symbol $\id$ denotes the identity irrep of the group $G$, $|U|=\operatorname{dim}U$, and $\mathbf{1}_{|U|}$ is the identity operator of dimension $|U|$.
\end{fact}

\begin{proof}
The first part follows from the Schur Lemma~\cite{FHa}, since we average $X$ over whole group, and $\forall g\in G \ \chi^{\operatorname{id}}(g^{-1})=1$. To prove the second part we have to apply one of the  Schur orthogonality relation~\cite{FHa}:
\begin{align}
\label{orr2}
\frac{1}{|G|}\sum_{g\in G}\chi ^{\alpha }\left( g^{-1}\right)&=
\begin{cases}
1 \ if \ \alpha =\id,\\ 
0 \ if~\ \alpha \neq \id.
\end{cases}
\end{align} 
\end{proof}

Finally, a linear map $\Phi \in \End[\M(d,\mathbb{C})]$ that is irreducible covariant with respect to unitary irreducible representation $U$ of a finite group $G$ admits the decomposition
\be
\label{ICQC}
\Phi=\sum_{\alpha \in \Theta}l_{\alpha}\Pi^{\alpha}, \quad l_{\alpha}\in \mathbb{C}.
\ee
The numbers $l_{\alpha}$ are eigenvalues of a covariant map $\Phi$, whenever the decomposition~\eqref{UUmult} is multiplicity-free.  If the multiplicities $m_{\alpha}\geq 1$ in~\eqref{UUmult}, then form~\eqref{ICQC} is still irreducibly covariant, but not the most general.  Having Fact~\ref{F3} and the explicit form of the $ICLMs$ in~\eqref{ICQC}, we prove the following:
\begin{fact}
	\label{f:decomp}
	For all $X\in \M(|U|, \mathbb{C})$, with $|U|=\operatorname{dim}U$, we have
	\begin{equation}
	X =\frac{ \text{Tr}(X)}{|U|}\mathbf{1}_{|U|} + \sum_{\alpha\neq{\operatorname{id}}}e_\alpha,
	\end{equation}
	where $e_\alpha\coloneqq\Pi^\alpha(X)\in \text{Im}(\Pi^\alpha), \;\text{Tr}(e_\alpha) = 0,\; e_\alpha e_\beta = \delta_{\alpha\beta}e_\beta$ and if $X$ is 
	hermitian then $e_\alpha$ is hermitian as well.
\end{fact}

\begin{proof}
	The statement follows directly from the form the definition of the rank-one projectors~\eqref{op_Pi}, their orthogonality and completeness~\eqref{orto_Pi}, and the previous Fact~\ref{F3}.
\end{proof}

It was shown in~\cite{Daniel}, that $l_{\alpha}$ can be chosen to be real if we restrict ourselves to $CPTP$ maps. 
To ensure that map $\Phi$ is trace-preserving ($TP$) we put $l_{\id}=1$, where $\id$ denotes the  trivial representation (see Proposition 25 in~\cite{Moz2017}). The complete positivity ($CP$) is obtained from the Choi-Jamiołkowski isomorphism, via theorem below.
	\begin{theorem}
		\label{mainTHM}
	The Choi-Jamio{\l}kowski isomorphism given through~\eqref{Ch-J} has the following properties:
	\begin{enumerate}
		\item A linear map is completely positive if and only if its Choi-Jamio{l}kowski image $J(\Phi)$ is a positive semidefinite matrix, i.e. $J(\Phi) \geq 0$.
		\item A linear map is positive but not completely positive if and only if its Choi-Jamio{l}kowski image is block-positive matrix, but not positive semifefinite. i.e. $\<x|\ot\<y|J(\Phi)|x\>\ot |y\>\geq 0$ for all vectors $|x\>,|y\>\in \mathbb{C}^d$, but there exists $|z\>\in \mathbb{C}^d\ot \mathbb{C}^d$ for which $\<z|J(\Phi)|z\><0$.
		\item A linear map is completely copositive ($CoP$) if and only if  Choi-Jamio{\l}kowski image of its composition with transposition $T$, i.e. $\Phi \circ T$ is positive semidefinite matrix.
	\end{enumerate}
	\end{theorem}

Thus, to have complete positivity, the coefficients $l_{\alpha}$ in expression~\eqref{ICQC} must satisfy the following set of inequalities~\cite{Moz2017}:
\begin{equation}
\label{sol}
\sum_{g\in G}\left( \sum_{\alpha \in \Theta }l_{\alpha }d_{\alpha} \chi
^{\alpha }\left( g^{-1}\right) \right) \left\vert \tr\left( V_{i}^{\beta
}U^{\dagger }(g)\right) \right\vert ^{2}\geq 0, \quad \forall \beta \in \Theta, \quad i \in \{1,\ldots, d_{\beta}\}.
\end{equation}%
In the above, $V_i^\beta \in \M(d_{\beta}, \cC)$ denote the normalized eigenvectors (see Proposition 20 in~\cite{Moz2017}) of rank one projectors $\Pi_i^\beta \in \End\left[\M( d_{\beta},\mathbb{C})\right]$ such that
$\Pi^\beta = \sum_i \Pi_i^\beta$. Expression~\eqref{sol} gives us the necessary and sufficient condition for a map in~\eqref{ICQC} to be CPTP whenever decomposition in~\eqref{UUmult} is multiplicity free. Whenever we consider a finite groups $G$ for which the numbers $m_{\alpha}$ in~\eqref{UUmult} can be greater than one, then we reduce to necessary conditions. This is due to the fact that if $m_{\alpha}>1$ for some irrep $\alpha$ in~\eqref{UUmult}, then the set $\Int_G(\Ad_U)$ is no longer spanned by the operators from~\eqref{op_Pi}.

\section{Positivite Irreducible Covariant Linear Maps}
\label{general_cons}
In the next three subsections, we discuss the conditions for a linear map admitting the decomposition from~\eqref{ICQC} to be an $ICLM$. In the first part, we do not restrict ourselves to finite groups having multiplicity-free decomposition~\eqref{UUmult}. Then the linear maps of the form~\eqref{ICQC} still belong to $\Int_G(\Ad_U)$, but in principle, we can find maps with the different decomposition than in~\eqref{ICQC}.   This, of course, gives us requirements for a specific but still wide class of linear maps satisfying condition of covariance from Definition~\ref{iclm}.  Moreover, such construction is universal and tractable in computations, allowing us for explicit constructions of new positive $ICLMs$. In the multiplicity-free case (low dimensions or monomial unitary group discussed later) we get very strong conditions for positivity since the form in~\eqref{ICQC} is the only allowed. The main result of this section is formulated in Theorem~\ref{th:pos_necess}. 

In the second part, we focus on the low dimensional case, where we work with multiplicity-free cases. In this regime, by straightforward calculations, we show that the necessary conditions obtained in the previous subsection are also sufficient. Additionally, in the third part, we briefly focus on co-positivity in the low dimensional case. In both, last subsections we focus on the permutation group $S(3)$ and quaternion group $Q$ with their two-dimensional irreps. We motivate this further discussion on qubit state transformations and Fujiwara-Algoet conditions.
\subsection{General Considerations}
\label{gen_cons2}
We can now formulate the full characterisation of a positive, but not completely positive map $\Phi\in \End[\M(d,\mathbb{C})]$ that is an $ICLM$  with respect to irreducible representation $U$ of a finite group $G$.  Let us observe that from the second point of Theorem~\ref{mainTHM} and the necessary form of $ICLM$ given in~\eqref{ICQC} we get the following 
\begin{theorem}
\label{th:P_ICLM}
Let the linear map $\Phi$ be ICLM with respect to unitary irreduible representation $U$ of finite group G, and thus it admits decomposition
\be
\Phi = \sum_\alpha l_\alpha\Pi^\alpha, \quad  l_\alpha \in \mathbb{C}.
\ee
Then it is positive, but not completely positive, when the spectral parameters $l_\alpha$ fulfil the following inequality
\be
\label{t_con}
\sum_{g\in G} \left(\sum_{\alpha \in \hat{G}}l_{\alpha}d_{\alpha}\chi^{\alpha}(g^{-1}) \right)\left|\tr\left(yx^TU^{\dagger}(g) \right)  \right|^2\geq 0,\quad \forall |x\>,|y\>\in \mathbb{C}^{\operatorname{dim}U}, 
\ee
where $\operatorname{dim}U$ is the dimension of the irrep $\alpha$ with respect to $ICLM$ is constructed for arbitrary $\ket{x}, \ket{y} \in \mathbb{C}^d$ and simultaneously violate at least one of the inequalities given in (\ref{sol}).
\end{theorem}
\begin{proof}
    Using the expression for the Choi-Jamio{\l}kowski image of an $ICLM \Phi$, provided in Proposition 26 in~\cite{Moz2017}
    \be
    J(\Phi) =\frac{1}{|G|} \sum_{i, j} E_{ij} \otimes \sum_{g} \left( \sum_{\alpha \in \hat{G}}l_\alpha d_\alpha\chi^\alpha(g^{-1})U_{C(i)}(g)U^\dagger_{R(j)}(g)\right),
    \ee
    where $U_{C(i)}$ denotes $i-th$ column of $U$, $U^\dagger_{R(j)}$ denotes $j-th$ row of $U^\dagger$ and $E_{ij} = \ket{i}\bra{j}$ for $\ket{i}, \ket{j}$ belonging to a given basis in $\mathbb{C}^d$,  we obtain by direct calculation that
    \be
    \<x|\ot\<y|J(\Phi)|x\>\ot |y\>\ = \sum_{g\in G} \left(\sum_{\alpha \in \hat{G}}l_{\alpha}d_{\alpha}\chi^{\alpha}(g^{-1}) \right)\left|\tr\left(yx^TU^{\dagger}(g) \right)  \right|^2.
    \ee
\end{proof}
This result, although fully describes necessary and sufficient condition for a given $ICLM$ to be positive, but not completely positive, turns out to be of small use in case of most $ICLMs$ (although in Section \ref{large_d} we can see that in one scenario it proved remarakbely efficient). Nevertheless, we can derive some neccessary condition, that can be easiliy applied.

If we consider the special case, when $|x\>=|i\>, |y\>=|j\>$, where $|i\>,|j\>$ belong to orthonormal basis of $\mathbb{C}^{d}$ we can formulate necessary condition for positivity of a $ICLM$ $\Phi\in \End[\M(d,\mathbb{C})]$. Namely we have the following
\begin{fact}
\label{f:P_ICLM}
If the $ICLM$ map $\Phi\in \End[\M(d,\mathbb{C})]$ is positive, then the elements $\Phi_{ii,jj}$ of $\operatorname{mat}(\Phi)$, where $1\leq i,j\leq \operatorname{dim} U$ satisfy
\be
\Phi_{ii,jj} = \frac{1}{|G|}\sum_{g \in G} \left(\sum_{\gamma \in \Theta}l_\gamma d_{\gamma}\chi^\gamma(g^{-1})\right) \left(U(g)\otimes\overline{U}(g)\right)_{ii,jj} \geq 0.
\ee
\end{fact}

\begin{proof}
Let ICLM $\Phi$ be positive and let $x = |j\>$ and $y=|i\>$ belong to some orthorormal basis of $\mathbb{C}^d$, and let $E_{ij} = \ket{i}\bra{j}$. Thus we have
\be
	\left| \tr\left( E_{ji}U^\dagger(g)\right)\right|^2 = \left| \sum_{k,l} \delta_{ki}\delta_{jl}U^\dagger_{lk}(g)\right|^2 = \left|U_{ji}^\dagger(g)\right|^2 = \left|U_{ij}\right|^2=\left(U(g)\otimes\overline{U}(g)\right)_{ii,jj}
\ee
and from \ref{t_con}
\begin{align}
0 \leq \sum_{g\in G} \left(\sum_{\alpha \in \hat{G}}l_{\alpha}d_{\alpha}\chi^{\alpha}(g^{-1}) \right)\left|\tr\left(yx^TU^{\dagger}(g) \right)  \right|^2 &=  \sum_{g\in G} \left(\sum_{\alpha \in \hat{G}}l_{\alpha}d_{\alpha}\chi^{\alpha}(g^{-1}) \right)\left| \tr\left( E_{ij}U^\dagger(g)\right)\right|^2 \\ &=\sum_{g\in G} \left(\sum_{\alpha \in \hat{G}}l_{\alpha}d_{\alpha}\chi^{\alpha}(g^{-1}) \right)\left(U(g)\otimes\overline{U}(g)\right)_{ii,jj} = \Phi_{ii,jj}.
\end{align}
\end{proof}



The inequalities in Theorem~\ref{th:P_ICLM} and Fact~\ref{f:P_ICLM}  are either technically difficult to solve or concern only a limited number of matrix elements of the matrix representation of a given ICLM.
(An analysis of  some geometrical properties of the solutions of
inequalities (23) is given in Section 7 of~\cite{Moz2017} and, what is
natural,  their solutions depend strongly on the group that defines
considered ICLM).
Surprisingly it appears that it possible to derive a universal 
necessary condition for positivity of ICLM, which is completely
independent from the group structure of the ICLM. The necessary
condition gives a geometrical restriction on the spectral parameters of
ICLM $\Phi$ and may be formulated in the following way.

\begin{theorem}\label{th:pos_necess}
	Let $\Phi$ be an $ICLM$ of the form
	\be
	\Phi=\sum_{\alpha \in \Theta}l_{\alpha}\Pi^{\alpha},\quad l_{\alpha}\in\mathbb{C},
	\ee
	described in~\eqref{ICQC}. Then if $\Phi$ is positive, the spectral parameters must obey \be|l_{\alpha}|\leq l_{\operatorname{id}}\quad \forall l_{\alpha}:\alpha \neq \operatorname{id}.\ee
\end{theorem}
\begin{proof}
	Let us  consider arbitrary hermitian $e_\alpha$  from the Fact \ref{f:decomp}. Since it is traceless (see Fact~\ref{F3}) its eigenvalues must fulfil $\lambda_M^\alpha > 0$ and $\lambda_m^\alpha<0$ (where $\lambda_M^\alpha$  and $\lambda_m^\alpha$ denote maximal and minimal eigenvalue, respectively). Now it is easy to check, that either $\lambda_M^\alpha \geq |\lambda_m^\alpha|$, or $\lambda_M^{\prime\alpha} \geq |\lambda_m^{\prime\alpha}|$ (where $\lambda^{\prime\alpha}$ is an eigenvalue of $e^{\prime}_{\alpha}\coloneqq-e_\alpha$ which belongs to $\text{Im}(\Pi^\alpha)$ as well). 
	
	This implies, that there exists $e_\alpha\in \text{Im}(\Pi^\alpha)$, such that $\lambda_M^\alpha \geq |\lambda_m^\alpha|$. Let us now consider 
	\begin{equation}
		f^\alpha\coloneqq |\lambda_m^\alpha|\mathbf{1} + e_\alpha \geq 0.
	\end{equation}
	Since $\Phi(f^\alpha) = l_{\operatorname{id}}|\lambda_m^\alpha|\mathbf{1} + l_\alpha e_\alpha$ we can clearly see, that for $l_\alpha > l_{\operatorname{id}}$ the map $\Phi$ cannot be positive. 
	Again, setting
	\begin{equation}
		f^{\prime\alpha}\coloneqq \lambda_M^\alpha\mathbf{1} + e_\alpha\geq0,
	\end{equation}
	we get $\Phi(f^{\prime\alpha}) = l_{\operatorname{id}}\lambda_M^\alpha\mathbf{1} + l_\alpha e_\alpha$ and thus, for $l_\alpha < -l_{\operatorname{id}},$ $ICLM$ $\Phi$ cannot be positive either.
\end{proof}
From the above theorem one can see that in the space of spectral parameters all the spectral parameters $l_{\alpha}$, such that $\alpha \neq \operatorname{id}$, must lie inside a cuboid whose edge lengths are equal to $2l_{\operatorname{id}}$. In particular $l_{\operatorname{id}}$ must be positive.

\subsection{Direct Approach to Positivity - Low Dimensions}
 In the following section, we discuss the case when the dimensions of the corresponding irreducible representation are low (two-dimensional). The general structure of every low-dimensional positive map is known due to works by Størmer and Woronowicz~\cite{Stromer,Woronowicz}. Namely, we know that every such map can be written as a sum of completely positive and completely co-positive maps. Nevertheless, it does not give us an explicit method of how to build positive maps, together with the mentioned decomposition, for a given finite group.
First, we focus on the case of two-dimensional irrep of the symmetric group $S(3)$ and then we move to appropriate irrep of the quaternion group $Q$. Besides, we present direct decomposition of maps into a sum of $CP$ and $CoP$ maps and show that components are also irreducibly covariant maps.
In these particular considerations, we set $l_{\id}=1$ in~\eqref{ICQC}, since we restrict to unital maps, i.e. preserving the identity operator.
The following calculations rely solely on the spectral decomposition of $ICLM \Phi$ given in \ref{ICQC} and the following  
\begin{fact}
\label{th:pos_proj}
	A linear map $\Phi :\mathbb{M}(n,\mathbb{C})\rightarrow \mathbb{M}(n,\mathbb{C})$ is positive if and only if $\Phi $ transforms rank one projectors into positive semidefinite
	matrices.
\end{fact}

\subsubsection{Positive irreducibly covariant linear maps for $S(3)$}
\label{subsubsub}
The permutation group $S(3)$ contains three irreducible representations: the identity representation $\phi^{\operatorname{id}}$, the sign representation $\phi^{\operatorname{sgn}}$, and non-trivial two dimensional representation $\phi^{\lambda}$.
Here we consider the positivity of $ICLM$ given in~\eqref{ICQC}, where instead of $U$ in~\eqref{op_Pi}, we take the two-dimensional irreducible representation $\phi^{\lambda}$.

Thanks to the Theorem \ref{th:pos_proj}, in order to check whether a map $\Phi :\mathbb{M}(n,\mathbb{C})\rightarrow \mathbb{M}(n,\mathbb{C})$ is positive we have to evaluate under what conditions the following holds
\be
\label{bla2}
\Phi (P)\geq 0:P=pp^{+}=\left( 
\begin{array}{cc}
|p_{1}|^{2} & p_{1}\overline{p}_{2} \\ 
\overline{p}_{1}p_{2} & |p_{2}|^{2}%
\end{array}%
\right) ,\quad p=\left( 
\begin{array}{c}
p_{1} \\ 
p_{2}%
\end{array}%
\right) ,\quad |p_{1}|^{2}+|p_{2}|^{2}=1.
\ee
Using the explicit form of $ICLM$ for the group $S(3)$, which is of the form $\Phi =\Pi ^{\id}+l_{\sgn}\Pi ^{\sgn}+l_{\lambda }\Pi^{\lambda }$, since $\phi^{\lambda}\ot \bar{\phi}^{\lambda}=\phi^{\operatorname{id}}\oplus \phi^{\operatorname{sgn}}\oplus \phi^{\lambda}$, hence~\eqref{bla2} reads as
\be
\label{PP}
\Phi (P)=\left( 
\begin{array}{cc}
\frac{1}{2}(1+l_{\sgn}(|p_{1}|^{2}-|p_{2}|^{2})) & l_{\lambda }p_{1}\overline{p}%
_{2} \\ 
l_{\lambda }\overline{p}_{1}p_{2} & \frac{1}{2}%
(1+l_{\sgn}(|p_{2}|^{2}-|p_{1}|^{2}))%
\end{array}%
\right) .
\ee
From this we deduce  that $l_{\sgn}$, $l_{\lambda }\in \mathbb{R}$, otherwise $\Phi (P)$ would not be hermitian. Now calculating the
eigenvalues of~\eqref{PP} one gets that the map $\Phi =\Pi ^{\id}+l_{\sgn}\Pi ^{\sgn}+l_{\lambda }\Pi ^{\lambda }$ is positive if and only if
	\be
	\label{pp3}
	1\geq l_{\sgn}^{2}(|p_{1}|^{2}-|p_{2}|^{2})^{2}+4l_{\lambda }^{2}|p_{1}|^{2}|%
	\overline{p}_{2}|^{2},
	\ee
	for any $p=\left( 
	\begin{array}{c}
	p_{1} \\ 
	p_{2}%
	\end{array}%
	\right), |p_{1}|^{2}+|p_{2}|^{2}=1.$
Now let us observe that $(|p_{1}|^{2}-|p_{2}|^{2|})^{2}+4|p_{1}|^{2}|\overline{p}_{2}|^{2}=1$, so we have 
\be
\label{asum1}
|p_{1}|^{2}-|p_{2}|^{2|}=\sin \theta ,\quad 2|p_{1}||\overline{p}_{2}|=\cos
\theta ,\quad \theta \in \lbrack 0,2\pi ). 
\ee
Using directly expression~\eqref{pp3} and expression~\eqref{asum1} we see that ma $\Phi$ is positive if and only if 
	\be
	\label{bla1}
	1\geq l_{\sgn}^{2}\sin ^{2}\theta +l_{\lambda }^{2}\cos ^{2}\theta ,
	\ee
	for any $\theta \in \lbrack 0,2\pi ).$ Above condition is equivalent to $|l_{\sgn}|\leq 1, |l_{\lambda }|\leq 1$. Indeed, if $|l_{\sgn}|>1$ or $|l_{\lambda }|>1$, then for,  $\theta =%
	\frac{\pi }{2}$ or $\theta =0$ the LHS of~\eqref{bla1} is violated. If $|l_{\sgn}|\leq 1,\quad |l_{\lambda }|\leq 1$, then we have 
	\be
	l_{\sgn}^{2}\sin ^{2}\theta +l_{\lambda }^{2}\cos ^{2}\theta \leq \sin
	^{2}\theta +\cos ^{2}\theta =1.
	\ee
Therefore we may formulate the main result for this paragraph:
\begin{theorem}
	The linear map $\Phi =\Pi ^{\id}+l_{\sgn}\Pi ^{\sgn}+l_{\lambda }\Pi ^{\lambda }$
	is positive if and only if $l_{\sgn}$, $l_{\lambda }\in \mathbb{R}$ and lie in the rectangle $|l_{\sgn}|\leq 1,\quad |l_{\lambda }|\leq 1.$
\end{theorem}

\subsubsection{Positive irreducibly covariant linear maps for quaternion group $Q.$}
Let us consider the quaternion group $Q$, defined by
\be
\label{defQ}
Q \coloneqq \<\pm e,\pm i,\pm j,\pm k:i^2=j^2=k^2=ijk=-e\>,
\ee
with its five irreducible representations: $\phi^{\id}$ - identity representation, $ \phi^{t_1},\phi^{t_2},\phi^{t_3}$ -sign representations and two dimensional irrep $\phi^{t_4}$.
 Then, the $ICLM$ $\Phi$ generated by the two-dimensional irreducible representation $\phi^{t_4}$, is of the form $\Phi=l_\id\Pi^\id + l_{t_1} \Pi^{t_1 }+ l_{t_2} \Pi^{t_2 }+ l_{t_3} \Pi^{t_3}$, since $\phi^{t_4}\ot \bar{\phi}^{t_4}=\phi^{\operatorname{id}}\oplus \phi^{t_1}\oplus \phi^{t_2}\oplus \phi^{t_3}$.  In the matrix representation the map $\Phi$ is of the form
\begin{align}\label{eq:iclm_q}
	\operatorname{mat}(\Phi) &= \frac{1}{2}\begin{pmatrix} l_{\id}+l_{t_2} & 0 & 0 & l_{\id}-l_{t_2}\\ 0 &l_{t_1}+l_{t_3}&l_{t_3}-l_{t_1}&0\\ 0 &l_{t_3}-l_{t_1}&l_{t_1}+l_{t_3}&0\\l_{\id}-l_{t_2} & 0 & 0 & l_{\id}+l_{t_2} \end{pmatrix}.
\end{align}
Due to Theorem~\ref{th:pos_proj}, restricting only to rank one projectors $P=pp^{+}$ for $p=(p_1,p_2)^t$, such that $|p_{1}|^{2}+|p_{2}|^{2}=1$, and setting $l_\id=1$, we have
\be
	\Phi (P)=\frac{1}{2}\left( 
	\begin{array}{cc}
	1+l_{t_2}(|p_{1}|^{2}-|p_{2}|^{2}) & l_{t_1}(p_{1}\overline{p}_{2}-\overline{p}%
	_{1}p_{2})+l_{t_3}(p_{1}\overline{p}_{2}+\overline{p}_{1}p_{2}) \\ 
	l_{t_1}(\overline{p}_{1}p_{2}-p_{1}\overline{p}_{2})+l_{t_3}(\overline{p}%
	_{1}p_{2}+p_{1}\overline{p}_{2}) & 1+l_{t_2}(|p_{2}|^{2}-|p_{1}|^{2})%
	\end{array}%
	\right) . 
	\ee
	Now taking $q=(q_1,q_2)^t$, such that $|q_{1}|^{2}+|q_{2}|^{2}=1$, we have for any normalised vectors $p,q$ $q^{+}\Phi (P)q\geq 0$  if and only if $%
	|l_{i}|\leq 1$, where  $i\in \{\operatorname{id}, t_1,t_2,t_3,t_4\}$. We can summarise our findings in the following:
\begin{theorem}
	\label{thmQQ}
	The map $\Phi=\Pi^\id + l_{t_1} \Pi^{t_1 }+ l_{t_2} \Pi^{t_2 }+ l_{t_3} \Pi^{t_3}$, generated by the two-dimensional irrep $\phi^{t_4}$ is positive if and only if $%
	|l_{i}|\leq 1$, for $i\in \{\operatorname{id}, t_1,t_2,t_3,t_4\}$.
\end{theorem}

\subsection{Remarks on Decomposability - Low Dimensions}

We would like to decompose an $ICLM$ $\Phi$, given by (\ref{eq:iclm_q}) into sum of completely positive map and completely co-positive map 
\begin{equation}\label{eq:rozklad}
\Phi = \Psi^{(1)} + \Psi^{(2)}\circ T,
\end{equation}
where $T$ denotes the transposition operator. In this section, we show that both $\Psi^{(1)}$ and  $\Psi^{(2)}\circ T$ have to be $ICLMs$ of the same structure as $\Phi$. 
First, let us formulate the following 
\begin{lemma}
    Let $\Phi$ be an $ICLM$ with respect to $\phi^{t_4}$ irrep of the quaternion group, given by (\ref{eq:iclm_q}), and  $J(\Phi)$ be its Choi-Jamio{\l}kowski image. Then if $\Phi$ is positive, but not completely positive then exactly one eigenvalue of $J(\Phi)$ is negative.
\end{lemma}
\begin{proof}
Examining the spectrum of $J(\Phi)$, one can see that the spectral parameters $l_{i}$, for $i\in \{\operatorname{id}, t_1,t_2,t_3,t_4\}$ of $ICLM$ with respect to the two-dimensional irrep $\phi^{t_4}$ of the quaternion group $Q$, given by (\ref{eq:iclm_q}), are connected with eigenvalues $\delta_i$ of the Choi-Jamiołkowski image of $\Phi$ as
\begin{align}\label{eq:delty}
\begin{pmatrix}\delta_{\id}\\ \delta_{t_1}\\ \delta_{t_2}\\ \delta_{t_3} \end{pmatrix}
&= I \begin{pmatrix}l_{\id}\\ l_{t_1} \\ l_{t_2} \\ l_{t_3} \end{pmatrix} = \frac{1}{2}\begin{pmatrix} 1&1&1&1 \\ 1&1&-1&-1 \\1&-1&1&-1\\1&-1&-1&1 \end{pmatrix} \begin{pmatrix}l_{\id}\\ l_{t_1} \\ l_{t_2} \\ l_{t_3} \end{pmatrix}, \quad I\neq \mathbf{1}.
\end{align} For an $ICLM$ that is positive we have $|l_i| \leq l_{\id}$ (from Theorem~\ref{th:pos_necess}), then writing all possible combinations $\delta_i + \delta_j, i\neq j$, i.e. 
\begin{align}
&\delta_{\id} + \delta_i = l_{\id} + l_i \geq 0\quad \forall i\in \{\operatorname{id}, t_1,t_2,t_3,t_4\},\\
&\delta_{i} + \delta_j = 2l_{\id} - 2l_k \geq 0\quad \forall i\neq j\neq k\neq \id,
\end{align}
one can observe that for $P$ but not $CP$ map $\Phi$, exactly one $\delta_i$ is negative. 
\end{proof}

Now we would like to decompose such $\Phi$ into sum of completely positive map and completely co-positive map like in (\ref{eq:rozklad}), 
where $\Psi^{(1)}$ and $\Psi^{(2)}\circ T$ are both CP $ICLMs$ with respect to the quaternion group of the form ($\ref{eq:iclm_q}$) and depend on spectral parameters $a_i, b_i, i\in \{\operatorname{id}, t_1,t_2,t_3,t_4\}$ (instead of $l_i$) respectively. Hence, since $I$ is involutive, they satisfy 

\begin{align}\label{eq:male}
[a_i]_i &= I [\varepsilon_i]_i,\quad  \varepsilon_i\geq0\;  \forall i,\\
[b_i]_i &= I [\gamma_i]_i,\quad \gamma_i\geq 0\; \forall i.
\end{align}

where $\delta_i, \varepsilon_i$ and $ \gamma_i$ denote the eigenvalues of respective Choi-Jamio{\l}kowski images of $\Phi, \Psi^{(1)}, \Psi^{(2)}$. Combining \eqref{eq:rozklad} with \eqref{eq:delty} and \eqref{eq:male} we obtain
\begin{align}\label{eq:rozkladd}
\begin{cases}
\delta_{\id} &= \varepsilon_{\id} + \frac{1}{2}(\gamma_{\id}-\gamma_{t_1}+\gamma_{t_2}+\gamma_{t_3}),\\
\delta_{t_1}& = \varepsilon_{t_1} + \frac{1}{2}(-\gamma_{\id}+\gamma_{t_1}+\gamma_{t_2}+\gamma_{t_3}),\\
\delta_{t_2} &= \varepsilon_{t_2} + \frac{1}{2}(\gamma_{\id}+\gamma_{t_1}+\gamma_{t_2}-\gamma_{t_3}),\\
\delta_{t_3} &= \varepsilon_{t_3} + \frac{1}{2}(\gamma_{\id}+\gamma_{t_1}-\gamma_{t_2}+\gamma_{t_3}).
\end{cases}
\end{align}
From the above equations one can see, that $\delta_i + \delta_j \geq 0$ if $i\neq j$. Without the loss of generality one can assume that $\delta_{\id}< 0$ (the exactly one negative eigenvalue mentioned earlier), thus $\delta_{\id}\geq|\delta_{\id}|$. Rewriting \eqref{eq:rozkladd} we get
\begin{align}
\begin{cases}
\delta_{\id}& = \varepsilon_{\id} + \frac{1}{2}(\gamma_{\id}-\gamma_{t_1}+\gamma_{t_2}+\gamma_{t_3}),\\
\delta_{t_1}&= - \delta_{\id} + \gamma_{t_2} + \gamma_{t_3} + \varepsilon_{\id} + \varepsilon_{t_1},\\
\delta_{t_2} &= - \delta_{\id} + \gamma_{\id} + \gamma_{t_2} + \varepsilon_{\id} + \varepsilon_{t_2},\\
\delta_{t_3} &= - \delta_{\id} + \gamma_{\id} + \gamma_{t_3} + \varepsilon_{\id} + \varepsilon_{t_3}.\\
\end{cases}
\end{align}
Since $\delta_i > |\delta_{\id}|$, we can write
\begin{align}
\begin{cases}
0 \geq \delta_{\id} = \varepsilon_{\id} + \frac{1}{2}(\gamma_{\id}-\gamma_{t_1}+\gamma_{t_2}+\gamma_{t_3}),\\
0 \leq \gamma_{t_2} + \gamma_{t_3} + \varepsilon_{\id} + \varepsilon_{t_1},\\
0 \leq  \gamma_{\id} + \gamma_{t_2} + \varepsilon_{\id} + \varepsilon_{t_2},\\
0 \leq \gamma_{\id} + \gamma_{t_3} + \varepsilon_{\id} + \varepsilon_{t_3}.\\
\end{cases}
\end{align}
For an arbitrary set of eigenvalues $\delta_i$ corresponding to $ICLM$ that is $P$ but not $CP$ the above equation has solution of the form
\begin{align}\label{eq:res_q}
\begin{pmatrix}\gamma_{\id}\\ \gamma_{t_1}\\ \gamma_{t_2}\\ \gamma_{t_3} \end{pmatrix} = \begin{pmatrix}0\\ -\delta_{\id}\\ 0\\ 0 \end{pmatrix} ,\quad
\begin{pmatrix}\varepsilon_{\id}\\ \varepsilon_{t_1}\\ \varepsilon_{t_2}\\ \varepsilon_{t_3} \end{pmatrix} = \begin{pmatrix}0\\ \delta_{\id} + \delta_{t_1}\\\delta_{\id} +  \delta_{t_2}\\\delta_{\id} +   \delta_{t_3}\end{pmatrix}.
\end{align}
Thus, we can formulate
\begin{theorem}
    An $ICLM \Phi$ with respect to the quaternion group irrep $\phi^{t_4}$ (given by (\ref{eq:iclm_q})) that is positive, but not completely positive can be decomposed into 
    \be
    \Phi = \Psi^{(1)} + \Psi^{(2)}\circ T,
    \ee
    where $\Psi^{(1)}$ and $\Psi^{(2)}$ are completely positive and are $ICLMs$ with respect to the irrep $\phi^{t_4}$ of the quaternion group as well.
\end{theorem}

When we set $l_{t_1}=l_{t_3}$ in \eqref{eq:iclm_q}, we obtain $ICLM$ with respect to the two-dimensional irrep $\phi^{\lambda}$ of the $S(3)$ group. Considering expression \eqref{eq:delty} we get that the equality must hold for appropriate Choi-Jamio{\l}kowski eigenvalues as well (i.e. $\delta_{t_1}=\delta_{t_3}$). Thus from equation~\eqref{eq:delty} and the fact that only one $\delta_i$, for $i\in \{\operatorname{id}, t_1,t_2,t_3,t_4\}$, can be negative we conclude, that for no choice of $\delta_i$ there is a solution for $(\gamma_i)$ such that $\Psi^{(2)}$ is $ICLM$ with respect  to $S(3)$ group. This, along with the discussion presented in Section~\ref{sec:unit_qc}, shows the special role quaternion of quaternion group $Q$ among the low-dimensional $ICLMs$.

\section{Positive Irreducibly Covariant Linear Maps from The Inverse Reduction Map}
\label{invRR}
 In this section, we present the efficient method allowing for the construction of irreducibly covariant positive linear maps using the inverse reduction map. In this approach we consider linear maps, not necessarily multiplicity free in the decomposition presented in~\eqref{UUmult} but admitting the decomposition given through~\eqref{ICQC}. The presented method allows us to construct $ICLMs$ also in higher dimensions, i.e. when the dimension of considered irrep is larger than 2. This is due to the fact, that this method rids of the direct checking the block positivity (see the second point of Theorem~\ref{mainTHM}) of the corresponding Choi-Jamio{\l}kowski image. Instead of that, we must ensure the positive spectrum of the certain hermitian operator, which is much more friendly in the practical applications.  Thanks to that, the presented approach is universal and can be applied to any linear map, not necessarily $ICLM$ one.
We start from the definition of the inverse reduction map. We start from the following:
\begin{definition}
	\label{invR}
	The inverse reduction map $R^{-1}\in \End[\M(d,\mathbb{C})]$ is given as
	\be
	\forall X\in \M(d,\mathbb{C}) \quad R^{-1}(X)=\frac{\tr(X)}{d-1}\mathbf{1}-X.
	\ee
\end{definition}
The reader can check by direct calculation that indeed map defined above is the inverse map with respect to the reduction map defined in~\cite{Hor2}. Moreover, the map $R^{-1}$ is surjective between set $\mathcal{P}_k^d$ of rank $k$ projectors and the set $\mathcal{P}_1^d$  of rank one projectors. Later in~\cite{ChrusWig} it has been shown that the map $R^{-1}$, from Definition~\ref{invR}, is the only one with such property. We will also need the following theorem proved in~\cite{Moz2015}:
\begin{theorem}
	\label{rank}
	Let $W\in \M(d^2,\mathbb{C})$ be a hermitian and non-positive operator, such that $\widetilde{W}=(\mathbf{1}\ot R^{-1})W\geq 0$. Then the operator $W$ is an entanglement witness (see Definition~\ref{def01}).
\end{theorem}

As was mentioned in the introduction in this manuscript we restrict ourselves to special class of linear maps which are irreducibly covariant, see expression~\eqref{irr_cov} in Section~\ref{S_ICLM}. Having the above, we prove the following:

\begin{lemma}
	\label{comp1}
	Let $R^{-1}\in \End[\M(d,\mathbb{C})]$ be the inverse reduction map given in Definition~\ref{invR} and $\Phi\in \End[\M(d,\mathbb{C})]$ be an $ICLM$ given through~\eqref{ICQC}, then we have
	\be
	\label{eq1}
	\forall X\in\M(d,\mathbb{C}) \quad R^{-1}\circ \Phi(X)=\frac{l_{\id}\tr(X)}{d-1}-\Phi(X),
	\ee
	and
	\be
	\label{eq2}
	\widetilde{W}:=(\mathbf{1}\ot R^{-1}\circ \Phi)P_d^+=\frac{l_{\id}}{d(d-1)}\mathbf{1}\ot \mathbf{1}-\frac{1}{d}J(\Phi),
	\ee
	where $P_d^+$ is projector on maximally entangled sate $|\psi^+_d\>=\frac{1}{\sqrt{d}}\sum_{i=1}^d|ii\>$, $J(\Phi)$ is the Choi-Jamio{\l}kowski image of $\Phi$ defined in~\eqref{Ch-J}, and the symbol $\circ$ means the composition of maps.
\end{lemma}
\begin{proof}
	First let us prove equation~\eqref{eq1}. For an arbitrary $X\in \M(d,\mathbb{C})$, we have
	\be
	\label{1}
	R^{-1}\circ \Phi(X):=\frac{1}{d-1}\tr\left[\Phi(X) \right]\mathbf{1}-\Phi(X). 
	\ee
	Using the general form of an $ICLM$, given in equation~\eqref{ICQC}, we compute
	\be
	\label{www}
	\begin{split}
		\tr\left[\Phi(X) \right]&=\sum_{\alpha \in \Theta}l_{\alpha}\tr\left[\Pi^{\alpha}(X) \right]=l_{\id}\tr(X).
	\end{split}
	\ee
	In the last equality we used Fact~\ref{f:decomp} from Subsection~\ref{gen_cons2}.
	Now, substituting~\eqref{www} into~\eqref{1} we get equation~\eqref{eq1}. Having the explicit expression for the composition $R^{-1}\circ \Phi$ from~\eqref{1}, we calculate $\widetilde{W}$ in~\eqref{eq2}:
	\be
	\begin{split}
		\widetilde{W}&:=(\mathbf{1}\ot R^{-1}\circ \Phi)P_d^+=\frac{1}{d}(\mathbf{1}\ot R^{-1}\circ \Phi)\sum_{ij}E_{ij}\ot E_{ij}=\frac{l_{\id}}{d(d-1)}\sum_{ij}E_{ij}\ot \mathbf{1}\delta_{ij}-\frac{1}{d}\sum_{ij}E_{ij}\ot \Phi(E_{ij})\\
		&=\frac{l_{\id}}{d(d-1)}\mathbf{1}\ot \mathbf{1}-\frac{1}{d}J(\Phi),
	\end{split}
	\ee
	since $\tr(E_{ij})=\delta_{ij}$.
\end{proof}
{Having an arbitrary linear map $\Phi$, not necessary irreducible, by defining $W:=(\mathbf{1}\otimes \Phi)P_d^+$, the conditions $\widetilde{W}\geq 0$ and $W<0$ imply the map $\Phi$ must be positive (see expression~\eqref{connP} and Theorem~\ref{rank}). One can see that requirement $W<0$ is nothing else but demanding that a map $\Phi$ is not completely positive. This general observation gives us a tool for systematic construction of positive maps among all $ICLMs$ considered in this manuscript by imposing constraints on the coefficients $l_{\alpha}$ (see expression~\eqref{dede}). To do so, first we need the following corollary giving the conditions for the positive spectrum of $\widetilde{W}$ given in Lemma~\ref{comp1}:}
\begin{corollary}
	\label{cor_Wtilde}
	From Lemma~\ref{comp1} follows  that $\widetilde{W}\geq 0$, if and only if 
	\be
	\label{cor_wtilde}
	\frac{l_{\id}}{d-1}-\epsilon^{\beta}_i\geq 0 \quad \beta \in \Theta, \quad i=1,\ldots,d_{\beta}.
	\ee
	The numbers $\epsilon^{\beta}_i$ are eigenvalues of the Choi-Jamio{\l}kowski image $J(\Phi)$ of an $ICLM$  $\Phi$ admitting the decomposition in~\eqref{ICQC}.
\end{corollary}
{Finally, Corollary~\ref{cor_Wtilde} and Theorem~\ref{rank} explain how to proceed to get an $ICLM$ which is positive but not completely positive:}
\begin{corollary}
	\label{Cor27}
	An $ICLM$ $\Phi \in \End[\M(d,\mathbb{C})]$ is positive when all inequalities given in~\eqref{cor_wtilde} of Corollary~\ref{cor_Wtilde} are satisfied and the map $\Phi$ is not completely positive, i.e. when at least one inequality from~\eqref{sol} is not fulfilled.
\end{corollary}

As an illustration of the above statements, let us consider the symmetric group $S(3)$ and $ICLM$ $\Phi$ with respect to the two-dimensional irrep labelled by $\lambda$ (see Subsection~\ref{subsubsub}). In this particular case, the conditions~\eqref{cor_wtilde} (left-hand side) and~\eqref{sol} (right-hand side) have a form
\be
\begin{cases}
	3l_{\id}-1\geq 0,\\
	l_{\id}+l_{\sgn}\geq 0,\\
	2l_{\lambda}+l_{\id}-l_{\sgn}\geq 0,\\
	-2l_{\lambda}+l_{\id}-l_{\sgn}\geq 0,
\end{cases}
\quad
\begin{cases}
	1-l_{\id}\geq 0,\\
	l_{\id}-l_{\sgn} \geq 0,\\
	l_{\sgn}+l_{\id}-2l_{\lambda}\geq 0,\\
	l_{\sgn}+l_{\id}+2l_{\lambda}\geq 0.
\end{cases}
\ee
Computing the Choi-Jamio{\l}kowski image $J(\Phi)$  and evaluating overlap with the maximally entangled state $P_+$, we have $\tr(P_+J(\Phi))=l_{\sgn}+2l_{\lambda}+l_{\id}$. We see that the point $(l_{\id},l_{\sgn},l_{\lambda})=(1,-1,-1)$ belongs to the region for which the map $\Phi$ is positive but not completely positive, and thus the $J(\Phi)$, for such a choice of parameters, is appropriate entanglement witness since it detects at least one entangled state. In further sections we apply the method developed here to construct other $ICLMs$ for different groups and in higher dimensions.
\section{New examples of positive (irreducibly) covariant  linear maps in higher dimensions}
\label{large_d}
In the following three subsections we present the analysis of the irreducibly covariant positive linear maps induced by three-dimensional irreducible representation of the permutation group $S(4)$ and $d-$dimensional representations of the group of monomial unitaries $MU(d)$, for an arbitrary $d$. In the case of the group $S(4)$ our results have been obtained using construction of a novel irreducible representation, see Appendix~\ref{NovelSn}, and methods developed in Section~\ref{invRR}. This approach allows for significant simplification of the respective Choi-Jamio{\l}kowski images comparing to the previous approach presented in~\cite{Moz2017} and based on the Young-Yamanouchi basis.  In the case of the group $MU(d)$, the necessary and sufficient conditions for positivity, for an arbitrary dimension are given. This allows us to compare with the method based on the inverse reduction map discussed in the previous section. Additionally, we show that when $d=3$ and the specific choice of parameters, the $ICLMs$ generated by irrep of $S(3)$ and $MU(3)$ are isomorphic. 
\subsection{Irreducibly covariant maps induced by symmetric group S(4)}
The permutation group $S(4)$ has a five irreducible representations. Two of dimension one: symmetric ($\operatorname{id}$), antisymmetric ($\operatorname{sgn}$), one two-dimensional ($\lambda_2$), and two three-dimensional labelled here by $\lambda_1$ and $\lambda_3$. If $\Phi^{\lambda_1}$ is the trace preserving $ICLM$ from~\eqref{ICQC}, with respect to three-dimensional irrep of $S(4)$ labelled by $\lambda_1$, it is of the form 
\begin{equation}
\label{S4}
	\Phi^{\lambda_1}=\Pi^{\id}+l_{\lambda_1}\Pi^{\lambda_1} + l_{\lambda_2}\Pi^{\lambda_2} + l_{\lambda_3}\Pi^{\lambda_3}.
\end{equation}
One can see that in the above decomposition we do not have the irrep $\operatorname{sgn}$. This is because in the three-dimensional space the multiplicity of $\operatorname{sgn}$ irrep is zero.
The corresponding Choi-Jamio{\l}kowski image, calculated in using construction introduced in Appendix~\ref{NovelSn}, has a form
\be
\label{chooi}
J(\Phi^{\lambda_1}) = \begin{pmatrix}
	a_1 & 0 & 0 & 0 & a_5 & 0 & 0 & 0 & a_7\\
	0 & a_3 & 0 & 0 & 0 & a_6 & 0 & 0 & 0\\
	0 & 0 & a_2 & 0 & 0 & 0 & a_8 & 0 & 0\\
	0 & 0 & 0 & a_3 & 0 & 0 & 0 & a_6 & 0\\
	a_5 & 0 & 0 & 0 & a_4 & 0 & 0 & 0 & a_5\\
	0 & a_6 & 0 & 0 & 0 & a_3 & 0 & 0 & 0\\
	0 & 0 & a_8 & 0 & 0 & 0 & a_2 & 0 & 0\\
	0 & 0 & 0 & a_6 & 0 & 0 & 0 & a_3 & 0\\
	a_7 & 0 & 0 & 0 & a_5 & 0 & 0 & 0 & a_1\\
	\end{pmatrix},
\ee
where 
\begin{align}
\begin{split}
a_1&=\frac{1}{6}(l_{\lambda_2}+3l_{\lambda_3}+2), \  a_2=\frac{1}{6}(l_{\lambda_2}-3l_{\lambda_3}+2),
a_3=\frac{1}{3}(1-l_{\lambda_2}), \ a_4=\frac{1}{3}(2l_{\lambda_2}+1),\\
a_5&=\frac{1}{2}(l_{\lambda_1}+l_{\lambda_3}), \ a_6=\frac{1}{2}(l_{\lambda_1}-l_{\lambda_3}),
a_7=\frac{1}{2}(l_{\lambda_1}+l_{\lambda_2}), \ a_8=\frac{1}{2}(l_{\lambda_1}-l_{\lambda_2}).
\end{split}
\end{align}
The eigenvalues of $J(\Phi^{\lambda_1})$ are of the form
\begin{align}
\label{map}
\epsilon^{\lambda_1}_1 &= \epsilon^{\lambda_1}_1=\epsilon^{\lambda_1}_1=\frac{1}{6}(2+3l_{\lambda_1}-2l_{\lambda_2} -3l_{\lambda_3}),\\
\epsilon^{\lambda_2}_1 &= \epsilon^{\lambda_2}_2= \frac{1}{6}(2-3l_{\lambda_1}+4l_{\lambda_2} -3l_{\lambda_3}),\\
\epsilon^{\lambda_1}_1 &= \epsilon^{\lambda_1}_1=\epsilon^{\lambda_1}_1=\frac{1}{6}(2-3l_{\lambda_1}-2l_{\lambda_2} +3l_{\lambda_3}),\\
\epsilon^{\id} &= \frac{1}{3}(1+3l_{\lambda_1}+2l_{\lambda_2} +3l_{\lambda_3}).
\end{align}
Thus, the conditions for $\Phi^{\lambda_1}$ to be positive, but not completely positive, based on the inverse reduction map, from Section~\ref{invRR} yields two regions, presented in  Figure \ref{fig:s4}.
\begin{figure}[h]
	\begin{centering}
		\includegraphics[width=0.4\textwidth]{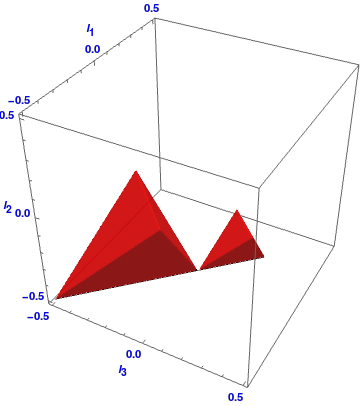}
		\caption{Visualisation of possible values of the parameters $l_{\lambda_1},l_{\lambda_2},l_{\lambda_3} $,  for which the map $\Phi^{\lambda_1}$, given in~\eqref{S4} is positive, but not completely positive. Regions obtained from inverse reduction map described in Section~\ref{invRR}.}
		\label{fig:s4}
	\end{centering}
\end{figure}
We can see, that the resulting region splits into two regions (larger and smaller), which are described respectively  by the following sets of inequalities
\begin{equation}
\begin{cases}
\frac{1}{2} - \frac{1}{6}(2-3l_{\lambda_1}-2l_{\lambda_2} +3l_{\lambda_3})\geq 0,\\
\frac{1}{2}-  \frac{1}{6}(2+3l_{\lambda_1}-2l_{\lambda_2} -3l_{\lambda_3}) \geq 0,\\
\frac{1}{2} - \frac{1}{6}(2-3l_{\lambda_1}+4l_{\lambda_2} -3l_{\lambda_3}) \geq0,\\
\frac{1}{3}(1+3l_{\lambda_1}+2l_{\lambda_2} +3l_{\lambda_3}) \leq 0,\\
\end{cases} \quad \text{and} \quad 
\begin{cases}
\frac{1}{2} - \frac{1}{6}(2-3l_{\lambda_1}-2l_{\lambda_2} +3l_{\lambda_3})\geq 0,\\
\frac{1}{2}-  \frac{1}{6}(2+3l_{\lambda_1}-2l_{\lambda_2} -3l_{\lambda_3}) \geq 0,\\
\frac{1}{2} - \frac{1}{3}(1+3l_{\lambda_1}+2l_{\lambda_2} +3l_{\lambda_3}) \geq 0,\\
\frac{1}{3}(1+3l_{\lambda_1}+2l_{\lambda_2} +3l_{\lambda_3}) \leq 0.\\
\end{cases}
\end{equation}
%
\subsection{Covariant positive maps induced by group of monomial unitaries}
\label{mon_sec}
\label{mon_sec}
As was stated in the introductory section, although in general the condition for block positivity of the  Choi-Jamio{\l}kowski image of a given map (Theorem~\ref{mainTHM}) is hard to work with, it proved efficient in a class of $d-$dimensional maps, covariant with respect to the subgroup of the monomial unitary group $MU(d)$.
Before we move to main considerations for this section, we start from the following
\begin{definition}
The group of monomial unitary matrices $MU(d)$ is given as the collection of unitaries $U\in U(d)$ of the form $U=DP$, where $D,P\in U(d)$ and $D$ is diagonal with respect to orthonormal basis of $\mathbb{C}^d$, and $P$ is a permutation matrix.
\end{definition}
Here, due to our further purposes we define subgroup $MU(d,n)$ of $MU(d)$ as
\begin{definition}
	\label{sMU}
We define $MU(d,n)$ to be the subgroup of the monomial unitary matrices of dimension $d$, whose non-zero entries consist only of $n-$th roots of unity. Since the natural representation of $MU(d)$ and $MU(d,n)$ are matrices, from now on saying about covariance with respect to $MU(d)$ or $MU(d,n)$ we mean covariance w.r.t. its matrix representation.
\end{definition}
A linear map $M$ covariant with respect to the subgroup $MU(d,n)$ from Definition~\ref{sMU}, for some $n\geq3$, is given by:
\begin{equation}
\label{M}
M(X) \equiv \text{Tr}(X) \frac{\mathbf{1}}{d} + \alpha \left(X - \sum_{k=1}^d\ket{k}\bra{k} \bra{k}X\ket{k}\right) + \beta\left(\sum_{k=1}^d\ket{k}\bra{k} \bra{k}X\ket{k} - \text{Tr}(X) \frac{\mathbf{1}}{d}\right),
\end{equation}
for every $X\in \mathbb{M}(d,\mathbb{C})$ and $\alpha,\beta \in \mathbb{R}$.
\begin{p_example}
\label{eq:mat_MU}
The matrix representation of the map $M\in\operatorname{End}(\mathbb{M}(3,\mathbb{C}))$ which is irreducibly covariant with respect to subgroup $MU(3, n)$ of monomial unitary group is given by
\begin{equation}
        \operatorname{mat}(M) =\left(
 \begin{array}{ccccccccc}
         \frac{1-\beta }{3}+\beta & 0 & 0 & 0&\frac{1-\beta }{3}  & 0 & 0 & 0&\frac{1-\beta }{3}  \\
0 & \alpha & 0 & 0 & 0 & 0 & 0 & 0 & 0 \\
0 & 0 & \alpha & 0 & 0 & 0 & 0 & 0 & 0 \\
0 & 0 & 0  & \alpha & 0 & 0 & 0 & 0& 0 \\
         \frac{1-\beta }{3} & 0 & 0&  0&\frac{1-\beta }{3}+\beta  & 0  & 0 & 0& \frac{1-\beta }{3} \\
0 & 0 & 0 & 0 & 0 & \alpha & 0 & 0 & 0 \\
         0 & 0 & 0 & 0 & 0 & 0 &  \alpha & 0  &0\\
         0 & 0 & 0 & 0 & 0 & 0 & 0 &  \alpha& 0 \\
         \frac{1-\beta }{3} & 0 & 0 &0 &\frac{1-\beta }{3}  & 0 &  0 & 0& \frac{1-\beta }{3}+\beta
\end{array}
\right)
\end{equation}

    \end{p_example}
\begin{fact}Let $M\in\operatorname{End}\mathbb{M}(3,\mathbb{C})$ be $ICLM$ with respect to $MU(3,n)$. Then the Choi-Jamio{\l}kowski image of $M$ is given by
\be
\label{JM}
\begin{split}
	J(M)&=
	\frac{\mathbf{1}\otimes\mathbf{1}}{d}+\alpha\left(\sum_{ij}E_{ij}\otimes E_{ij} - \sum_i E_{ii}\otimes E_{ii}\right) + \beta\left(\sum_i E_{ii}\otimes E_{ii} - \sum_i\frac{E_{ii}\otimes\mathbf{1}}{d}\right).
	\end{split}
\ee
\end{fact}
The expression above allows us to examine complete positivity and, which is more remarkable, positivity as well.
\begin{theorem}
The eigenvlues of $J(M)$ from~\eqref{JM} are given by
\begin{equation}
	\epsilon_1  = \frac{1}{3}{(1-\beta) },\quad\epsilon_2  = \frac{1}{3} (-3 \alpha +2 \beta +1),\quad \epsilon_3  = \frac{1}{3} (6 \alpha +2 \beta +1)
\end{equation}
and thus $M$ is completely positive if and only if $\alpha, \beta$ fulfil

\be
\begin{cases}
\label{e:M_cp}
	\frac{1}{3}{(1-\beta)} \geq  0,\\
	 \frac{1}{3} (-3 \alpha +2 \beta +1) \geq 0,\\
	\frac{1}{3} (6 \alpha +2 \beta +1) \geq 0.
\end{cases}
\ee
\end{theorem}
We can calculate the expectation value of $J(M)$ on arbitrary product vector $|x\> \ot |y\>\in \mathbb{C}^d\ot \mathbb{C}^d$:
\begin{align}\label{eq:block_pos}
	\begin{split}
		\bra{x}\otimes\bra{y}J(M)\ket{x}\otimes\ket{y} = \frac{1}{d} \biggl(&2d\alpha\sum_{i<j}x_iy_ix_jy_j +(1-\beta) \sum_{i\neq j} x_i^2y_j^2 +(1+(d-1)\beta) \sum_{i}x_i^2y_i^2\biggr),
	\end{split}
\end{align}
where $\ket{x}=(x_1, x_2, x_3),\ket{y}=(y_1, y_2, y_3)$. 

We can use expression~\eqref{eq:block_pos} to describe the whole region of positivity of $ICLM$ for $M$. First, in Lemma~\ref{lem:ex_points} and Theorem~\ref{th:suff} we describe sufficient conditions regarding $\alpha, \beta$, for which $M$ is positive. Then, in Theorem~\ref{thm:nec} we show that the sufficient region is neccesary as well, which gives the whole description of positive $ICLMs$ with respect to monomial unitary subgroup in arbitrary dimension.
\begin{lemma}\label{lem:ex_points}
	For the following values of $\alpha, \beta$ the map $M$, given through~\eqref{M}, is positive:
	\begin{eqnarray}
		\begin{cases}
			\alpha=1,\quad \beta=1,\\
			\alpha = -\frac{1}{d-1},\quad \beta=1,\\
			\alpha=\frac{1}{d-1},\quad \beta=-\frac{1}{d-1},\\
			\alpha=-\frac{1}{d-1},\quad \beta=-\frac{1}{d-1}.
		\end{cases}
	\end{eqnarray}
\end{lemma}
\begin{proof}
	Substituting the above values of $\alpha$ and $\beta$, one obtains the following expressions for $\bra{x}\otimes\bra{y}J(M)\ket{x}\otimes\ket{y}$ from~\eqref{eq:block_pos}:
\be
	\begin{cases}
	\displaystyle\frac{1}{d} \left( \sum_i x_iy_i\right) ^2,\quad \text{for } \alpha=1,\; \beta=1,\\
	\displaystyle\frac{1}{d-1}\sum_{i<j}(x_iy_i-x_jy_j)^2,\quad  \text{for } \alpha=-\frac{1}{d-1},\; \beta=1,\\
	\displaystyle\frac{1}{d-1}\sum_{i\neq j}(x_iy_j + x_jy_i)^2\quad  \text{for } \alpha=\frac{1}{d-1},\; \beta=-\frac{1}{d-1},\\
	\displaystyle\frac{1}{d-1}\sum_{i\neq j}(x_iy_j - x_jy_i)^2\quad  \text{for } \alpha=-\frac{1}{d-1},\; \beta=-\frac{1}{d-1},
	\end{cases}
\ee
	which are obviously positive for an arbitrary choice of $\ket{x}\otimes{\ket{y}}\in \mathbb{C}^d\ot \mathbb{C}^d$.
\end{proof}
Now we are able to formulate two main results of this section.
\begin{theorem}\label{th:suff}
	For the all values $\alpha, \beta$ within the quadrilateral region spanned by the points given in Lemma \ref{lem:ex_points}, the map $M$ is positive.
\end{theorem}
\begin{proof}
Let $\beta=-\frac{1}{d-1}$, $\alpha \in [-\frac{1}{d-1}, \frac{1}{d-1}]$. Expression \eqref{eq:block_pos} is positive for an arbitrary  $\ket{x}\otimes\ket{y}$, when $\beta=-\frac{1}{d-1}$ and either $\alpha=-\frac{1}{d-1}$ or $\alpha=\frac{1}{d-1}$. On the other hand, equation \eqref{eq:block_pos} is linear in $\alpha$ for fixed $\beta$ and $\ket{x}\otimes\ket{y}$. Thus, for all $\alpha\in [-\frac{1}{d-1}, \frac{1}{d-1}]$ and $\beta=-\frac{1}{d-1}$, the matrix $J(M)$ is block-positive, and thus the map $M$ is positive.
	Furthermore, let $\alpha\in[\frac{1}{d-1}, 1]$ and $\beta = \frac{d}{d-2}\alpha-\frac{2}{d-2}$. Then, writing for the convenience $\alpha=\frac{d-2}{d}\beta+\frac{2}{d}$, we get
	\be
	\begin{split}
		\bra{x}\otimes\bra{y}J(M)\ket{x}\otimes\ket{y}  = \frac{1}{d}\biggl(&2d\left(\frac{d-2}{d}\beta+\frac{2}{d}\right)\sum_{i>j}x_iy_ix_jy_j + (1-\beta) \sum_{i\neq j}x_i^2y_j^2 +(1+(d-1)\beta)\sum_ix_i^2y_i^2)\biggr) \\
		=\frac{1}{d}\biggl(&(1-\beta)\sum_{i>j} (x_iy_j+x_jy_i)^2+(1+(d-1)\beta)\left(\sum_ix_iy_i\right)^2\biggr),
	\end{split}
	\ee
	which is again clearly positive for $\beta \in [-\frac{1}{d-1}, 1]$.
	Using the same argument as for the $\beta=-\frac{1}{d-1}, \alpha = \pm \frac{1}{d-1}$, we can show that for $\beta=1, \alpha\in[-\frac{1}{d-1},1]$ expression \eqref{eq:block_pos} is also positive. Finally, fixing $\alpha$ and considering linearity of expression \eqref{eq:block_pos} in $\beta$ we obtain positivity for the whole region.
\end{proof}

\begin{theorem}
\label{thm:nec}
	The region described in Theorem \ref{th:suff} is not only sufficient, but necessary as well.
\end{theorem}
\begin{proof}
	To show most of the necessary conditions for positivity of map $M$ we choose special $\ket{x}$ and $\ket{y}$, presented in Table \ref{tab:n} along with the condition yield.
\begin{table}[h]
	\centering
\begin{tabular}{|c|c|c|c|c|}
	\hline
	$\ket{x}$ & $\ket{y} $  & $\bra{x}\otimes\bra{y}W_T\ket{x}\otimes\ket{y} $ &  condition\\ \hline
	$\sum_i \ket{i}$ & $\sum_i \ket{i}$  & $d(1+ (d-1)\alpha)$  &  $\alpha \geq -\frac{1}{d-1}$  \\
	$-\ket{1}+\ket{2}$ & $\sum_i \ket{i}$  & $(d-1)(1-\alpha)$  &  $\alpha \leq  1$  \\
	$\ket{1}$ & $\ket{1}$ &  $\frac{1+(d-1)\beta}{d}$ & $\beta\geq-\frac{1}{d-1}$   \\ 
	$\ket{1}$ & $\ket{2}$  & $\frac{1-\beta}{d}$  & $\beta\leq 1$   \\ \hline
\end{tabular}
	\caption{Derivation of the necessary conditions of the positivity of the map  $M$, basing on special choice of $\ket{x}, \ket{y}$ where $\{\ket{i}\}$ is the standard basis in $\mathbb{C}^d$.}\label{tab:n}
\end{table}

The final condition that makes the necessary and sufficient regions coincide is obtained by taking  $\beta < \frac{d}{d-2}\alpha-2/(d-2)$.  Let us set
\begin{equation}
	\alpha= \frac{d-2}{d}\beta + 2/d +\frac{\varepsilon}{2d}, \quad \varepsilon>0.
\end{equation}
Thus, condition for block positivity becomes
\be
	\bra{x}\otimes\bra{y}J(M)\ket{x}\otimes\ket{y} = \frac{1}{d}\left(   (1-\beta)\left( \sum_{i< j}(x_iy_j+x_jy_i)^2\right)  + (1+(d-1)\beta)\left( \sum_ix_iy_i\right) ^2+\varepsilon\sum_{i< j} x_iy_ix_jy_j\right) .
\ee
Clearly, this expression can be negative, e.g. let us take $\ket{x}=-(\ket{1}+\ket{2})$ and $\ket{y}=\ket{2}-\ket{1}$, then
\begin{equation}
	\bra{x}\otimes\bra{y}W_M\ket{x}\otimes\ket{y}=-\varepsilon.
\end{equation}
\end{proof}
One can compare the result above with the one obtained from the inverse reduction map.
\begin{fact}
\label{f:M_red}The region for which $ICLM$ $M$ is positive but not completely positive, obtained from inverse reduction map is given by the following set of inequalities
\begin{align}
\begin{cases}
	\frac{1}{d-1} - \frac{1-\beta }{d} \geq 0,\\
	\frac{1}{d-1} - \frac{1}{d} (-d \alpha +(d-1) \beta +1) \geq 0,\\
	\frac{1}{d-1} - \frac{1}{d} (d(d-1) \alpha +(d-1) \beta +1) \geq 0,
\end{cases}
\end{align}
while at least one inequality given in (\ref{e:M_cp}) is violated.
\end{fact}
The comparison of all the regions discussed above is presented in Figure~\ref{fig:PnCP}.
\begin{figure}
\subfloat[][]{\includegraphics[width=0.5\textwidth]{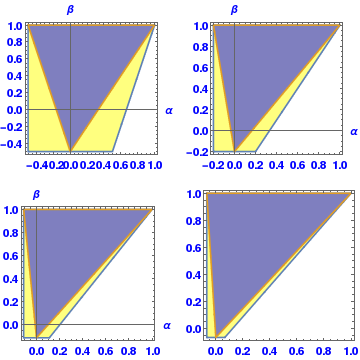}}
\subfloat[][]{\includegraphics[width=0.5\textwidth]{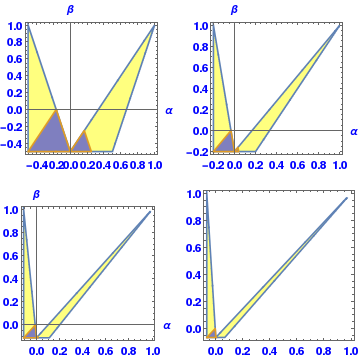}}
	\caption{In the subplot (a) we present possible values for the parameters $\alpha, \beta$ for $d=3,4,6,10$ for which the map $M$ from~\eqref{M} is either positive (yellow) or completely positive (blue). In the subplot (b) we compare general solution (yellow) for the P but not CP region, with the one obtained from the inverse reduction map (blue).  One observes that for larger dimensions $d$ the region for positivity (without complete positivity) shrinks.}
		\label{fig:PnCP}
\end{figure}
\FloatBarrier

\subsection{Connection between maps induced by group of monomial unitaries and group $S(4)$}
One can observe that the solution for positivity for the map $M$, for $d=3$ obtained from the inverse reduction map resembles the solution for $S(4)$ (see Figure \ref{fig:s4}) taken along the plane $l_{\lambda_1}=l_{\lambda_3}=\alpha, l_{\lambda_2}=\beta$.  What is more, for such choice of the parameters, all eigenvalues of $J(\Phi^{\lambda_1})$ and $J(M)$ are the same. Since both the matrices are hermitian, if they have common spectra they are similar. Moreover, using properties of  isomorphism, $M$ and $\Phi^{\lambda_1}$ (see~\eqref{S4}) have to be similar as well.
The connection between the entanglement witnesses from both groups is as follows

\begin{fact}
Let $M\in \operatorname{End}(\mathbb{M}(\mathbb{C},3)$ be an $ICLM$ with respect to monomial unitary subgroup $MU(3, n)$ and let $\Phi^{\lambda_1} \in \operatorname{End}(\mathbb{M}(\mathbb{C},3)$ be ICLM with respect to symmetric group $S(4)$. Then for $\alpha, \beta$ described in Fact~\ref{f:M_red} their respective Choi-Jamio{\l}kowski images related by
\be
J(M(\alpha,\beta))=AJ(\Phi^{\lambda_1}(\alpha, \beta))A^{-1}, 
\ee
where the matrix $A$ is given by
\begin{equation}
A=\left(
\begin{array}{ccccccccc}
 0 & 0 & \frac{1}{3} & 0 & 1 & 0 & -\frac{1}{3} & 0 & 0 \\
 0 & 1 & 0 & 0 & 0 & 0 & 0 & 0 & 0 \\
 0 & 0 & 0 & 1 & 0 & 0 & 0 & 0 & 0 \\
 0 & 0 & 0 & 0 & 0 & 1 & 0 & 0 & 0 \\
 0 & 0 & -\frac{2}{3} & 0 & 1 & 0 & \frac{2}{3} & 0 & 0 \\
 0 & 0 & 1 & 0 & 0 & 0 & 1 & 0 & 0 \\
 0 & 0 & 0 & 0 & 0 & 0 & 0 & 1 & 0 \\
 -1 & 0 & 0 & 0 & 0 & 0 & 0 & 0 & 1 \\
 1 & 0 & \frac{1}{3} & 0 & -1 & 0 & -\frac{1}{3} & 0 & 1 \\
\end{array}
\right).
\end{equation}
\end{fact}
\section{Comparison with the generalized Choi map}
\label{Choi-porownanie} 
The Choi map \cite{cho} and its generalised version~\cite{choi2} are examples of positive, but not completely positive endomorphisms of $\M(3, \mathbb{C})$. They are discussed, either in non-normalized \cite{cho} or normalized form \cite{chru_kos}. Since the $ICLM$ with respect to the monomial unitary subgroup is trace-preserving, we shall compare it to the latter one. The normalised and generalised version of the Choi map is defined by its action on the standard operator basis $\{e_{ij}\}_{i,j=1}^3$, with parameters $a,b,c \geq 0$, as follows:
\be
\begin{split}
	\Lambda[a,b,c](e_{ii}):=\sum_{j=1}^3d_{ij}e_{jj},\quad
	\Lambda[a,b,c](e_{ij}):=-\frac{1}{a+b+c}e_{ij}\quad \text{for} \quad i\neq j,
\end{split}
\ee
where the coefficients $d_{ij}$, for $1\leq i,j\leq 3$ form a matrix:
\be
d=\frac{1}{a+b+c}\begin{pmatrix}
	a & b & c\\
	c & a & b\\
	b & c & a
\end{pmatrix}.
\ee
Using the above, together with equation~\eqref{Ch-J}, we compute the Choi-Jamio{\l}kowski image of $\Lambda[a,b,c]$:
\be
\label{gchoi}
J(\Lambda[a,b,c])=\frac{1}{a+b+c}\begin{pmatrix}
	a & 0 & 0 & 0 & -1 & 0 & 0 & 0 & -1\\
	0 & b & 0 & 0 & 0  & 0 & 0 & 0 & 0\\
	0 & 0 & c & 0 & 0  & 0 & 0 & 0 & 0\\
	0 & 0 & 0 & c & 0  & 0 & 0 & 0 & 0\\
	-1& 0 & 0 & 0 & a  & 0 & 0 & 0 & -1\\
	0 & 0 & 0 & 0 & 0  & b & 0 & 0 & 0\\
	0 & 0 & 0 & 0 & 0  & 0 & b & 0 & 0\\
	0 & 0 & 0 & 0 & 0  & 0 & 0 & c &0\\
	-1& 0 & 0 & 0 & -1 & 0 & 0 & 0 &a 
\end{pmatrix}.
\ee
We know that the map $\Lambda[a,b,c]$ is \textit{non-decomposable and positive} \cite{chru_kos} if
\begin{enumerate}[a)]
	\item 
	\be
	a\leq 2,
	\ee
	\item 
	\be
	a+b+c\geq 2,
	\ee
	\item 
	\be
	\label{3}
	\begin{cases}
		(1-a)^2\leq bc \leq (2-a)^2/4 \quad \text{for} \quad 0\leq a\leq 1,\\
		0\leq bc\leq (2-a)^2/4 \quad \text{for} \quad 1\leq a\leq 2.
	\end{cases}
	\ee
\end{enumerate}
To have conditions for the map $\Lambda[a,b,c]$ to be \textit{positive}  but not completely positive, we have to use instead of~\eqref{3} the following constraints~\cite{chru_kos2}:
\be
\text{if} \quad a\leq 1 \quad \text{then} \quad bc\geq (1-a)^2.
\ee
\paragraph{Comparison with the generalized Choi map - the permutation group $S(4)$.}
The Choi-Jamio{\l}kowski image of $J(\Phi)$ from~\eqref{chooi} can be  equal to $J(\Lambda[a,b,c])$ from~\eqref{gchoi}, only when we put $l_{\lambda_3}=l_{\lambda_2}=l_{\lambda_1}=l$ in equation~\eqref{map}. This implies the following conditions for the parameters $a,b,c$:
\be
\label{equ}
a=-\frac{2l+1}{3l}, \ b=c=-\frac{1-l}{3l}.
\ee
\paragraph{Comparison with the generalised Choi map - the subgroup of monomial unitaries $MU(n,d)$.}
Comparing the matrix $\Lambda[a,b,c]$  with the Choi-Jamio{\l}kowski image $J(M)$, for $d=3$, we arrive at the following set of equations
\begin{equation}
	\label{eq:ab_al_be}
\begin{cases}
\frac{-1}{a+b+c}=\alpha,\\
\frac{a}{a+b+c}=\frac{1-\beta}{3} +\beta,\\
\frac{b}{a+b+c}=\frac{1-\beta}{3} ,\\
\frac{c}{a+b+c}=\frac{1-\beta}{3}.\\
\end{cases}
\end{equation}
From~\eqref{eq:ab_al_be}, we can immediately see that for positive $\alpha$ both pictures cannot coincide. Thus,  for $\alpha>0$, we obtain the positive and non-completely positive map, which is different from the generalised Choi map described in the literature.
Another immediate observation is that for $\alpha<0$ both maps coincide if $b=c$. This simplifies the requirement for positivity and non-complete positivity to
\begin{align}
	\label{eq:ab_p_ncp}
	\begin{cases}
&a\leq 2,\\
&a+2b\geq2,\\
&1-a\leq b, a \in [0, 1].
	\end{cases}
\end{align}
Solving equation \eqref{eq:ab_al_be} and taking into account condition \eqref{eq:ab_p_ncp}, one check that for every $(\alpha, \beta)$, in the remaining region (except for one point, where $\alpha=0$), we have the following
\begin{fact}
	In the region presented in Figure \ref{fig:final}, corresponding to $\alpha<0$, the $ICLM$ $M(\alpha, \beta)$ can be expressed as the generalised Choi map $\Lambda(a,b,b)$ which is positive, but not completely positive. Or conversely, the map $\Lambda(a,b,b)$ can be expressed as $M(\alpha, \beta)$.
\end{fact}

Both regions,i.e. for $\alpha<0$ and $\alpha>0$ are presented in  Figure \ref{fig:final}.
\begin{figure}[h]
	\begin{centering}
		\includegraphics[width=0.4\textwidth]{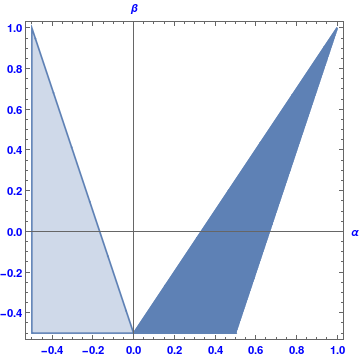} 
		\caption{The general region for positive but not completely positive $ICLMs$ generated by the group $MU(3)$. The blue triangle on the right-hand side of the point $(0,0)$ represents the solutions different from the generalized Choi map. It is worth noting that the region of overlap between the generalized Choi map and $MU(3)$ is not closed since we have to exclude the point $(0,0)$.}
		\label{fig:final}
	\end{centering}
\end{figure}

\paragraph{Decomposability of $ICLMs$ generated by the monomial unitary group.}
Since the $ICLMs$ induced by the monomial unitary group coincides with the generalized Choi map, we can examine the conditions for positivity and non-decomposability. By~\cite{chru_kos} we know that the map $\Lambda(a,b,c)$ is non-decomposable if:
\begin{align}
	&0\leq a\leq2,\\
	&a+b+c\geq 2,\\
	&\begin{cases}
	(1-a)^2 \leq bc \leq (2-a)^2/4 \;\text{for}\; 0\leq 1,\\ 
	0 \leq bc \leq (2-a)^2/4 \;\text{for}\; 1 \leq a \leq 2.
\end{cases}
\end{align}
Since for the region of coincidence of these two maps, we have $\alpha = -1/(a+b+c)$ and $\alpha \geq-1/2$, we can see that in the region of overlap the $ICLM$ is decomposable. A natural question arises, if the explicit decomposition can be pointed. Unfortunately, no non-trivial results were obtained in this respect.

The decomposition of the following form
\begin{equation}
M(\alpha, \beta) = pM^{(1)}(\alpha_1, \beta_1) + (1-p)M^{(2)}(\alpha_2, \beta_2) \circ T\quad p\in[0,1]
\end{equation}
where $M, M^{(1)}, M^{(2)}$ are of the form (\ref{M}) and $M$ is positive, but not completely positive and $M^{(1)}, M^{(2)}$ are completely positive does not have any solutions except for the trivial case (i.e. $p=0$).
Similarly, the decomposition 
\be
M(\alpha, \beta) = pM^{(1)}(\alpha_1, \beta_1) + (1-p)M^{(2)}(\alpha_2, \beta_2) \quad p\in[0,1]
\ee
where $M(\alpha_2, \beta_2) = \Phi(\alpha_2, \beta_2)\circ T$ and $M^{(1)}$ and $\Phi(\alpha_2, \beta_2)$ are CP is not contradictory only if $p = 0$. In such case $M$ is completely co-positive. From the examination of $J(\Phi)$ one can describe the region of complete co-positivity of M as follows: 
\begin{fact}
A map $M\in \operatorname{End}(\mathbb{M}(\mathbb{C}, 3)$, given by  (\ref{M}) is completely co-positive if the following inequalities hold
\be
\begin{cases}
1/3 (1 - 3 \alpha - \beta)\geq 0,\\
1/3 (1 + 3 \alpha - \beta)\geq 0,\\
1/3 (1 + 2 \beta)\geq 0.
\end{cases}
\ee
\end{fact}

\section{Irreducibly covariant linear maps and unital quantum channels}
\label{sec:unit_qc}
To illustrate the importance of the general studies on $ICLMs$, below we investigate a novel connection between all qubit unital quantum channels and linear maps generated by the irreducible representations of the quaternion group $Q$. 
To understand this connection, first, we have to present two methods of description of a qubit quantum state transformations. A qubit state is a matrix 
\be
\rho =(\rho _{ij})\in \mathbb{M}(2,\mathbb{C}):\tr(\rho )=1,\quad \rho \geq 0,
\ee
and the linear maps $\Phi \in \operatorname{End}[\mathbb{M}(2,\mathbb{C})]$ acts directly on the state matrices $\Phi :\rho \rightarrow \rho
^{\prime }$. There is also an alternative method of description of qubit states transformation based
on the presentation of the quantum states $\rho $ on the Bloch sphere%
\begin{equation}
\label{kubit}
\rho \equiv \rho (r)=\frac{1}{2}(\mathbf{1}_{2}+\sum_{i=1}^{3}r_{i}\sigma
_{i}):\sum_{i=1}^{3}r_{i}^{2}\leq 1,\quad r=(r_{i})\in \mathbb{R}^3,
\end{equation}
where $\sigma_{i}$ are Pauli matrices. Here, due to the properties of the
Pauli matrices we have a one-to-one correspondence between the states $\rho
\in \mathbb{M}(2,\mathbb{C})$ and the vectors $r=(r_{i})\in \mathbb{R}^{3},$ so we may describe the state transformations as the transformations $%
T\in \operatorname{End}[\mathbb{R}^{3}]\equiv \mathbb{M}(3,\mathbb{R})$. From the paper of K{\"u}bler and Braun~\cite{Kubler} we know that the most general
form of the unital qubit quantum channel $\Phi^T[\rho(r)]$ is 
\begin{equation}
\label{TT}
\Phi^T[\rho(r)]=\rho(Tr),
\end{equation}
and is fully characterized by $T\in \text{End}[\mathbb{R}^3]$ given by
\begin{equation}
\label{T}
T\equiv T(\eta ,R_{1},R_{2})=R_{1}D(\eta )R_{2}:R_{i}\in SO(3),\qquad D(\eta
)=\operatorname{diag}(\eta _{1},\eta _{2},\eta _{3})\in \mathbb{M}(3,
\mathbb{R}).
\end{equation}
The most important quantum properties of the map $T$, like positivity
and complete positivity  depend on the matrix $D(\eta )$ only, not on
the matrices\ $R_{i}\in SO(3)$, so the unital quantum channel is completely
characterized by the numbers $(\eta _{1},\eta _{2},\eta _{3})$, called in~\cite{Kubler} signed singular values (SSV).

From~\cite{Kubler} thr map $\Phi^T$ is completely positive, if and only if 
\begin{equation}
1+\eta_3\geq |\eta_1+\eta_2|,\quad 1-\eta_3\geq|\eta_1-\eta_2|.
\end{equation}
Now, let us consider the non-abelian quaternion group $Q$, described in Section~\ref{S_ICLM} (see also~\cite{Moz2017}).
Taking the $ICLM$ generated by the irrep $\phi^{t_4}$, with $l_{\operatorname{id}}=1$ in~\eqref{eq:iclm_q}, the conditions for $CP$ are
%
\begin{equation}
\label{CPQ}
1+l_{t_2}\geq |l_{t_1}+l_{t_3}|, \quad 1-l_{t_2}\geq |l_{t_1}-l_{t_3}|.
\end{equation}
These are exactly the same conditions as the conditions for $CP$ of the unital qubit quantum channel $\Phi^T$~\eqref{TT}, defined through  $T(\eta,R_1,R_2)$ in~\eqref{T}. In particular, if $R_1=R_2=\mathbf{1}$, the operators $\Phi^T(\eta, \mathbf{1}, \mathbf{1})$  and $\Phi$ generated for $\phi^{t_4}$ coincide. We summarise this in the following
\begin{corollary}
	Every "diagonal" unital quantum channel $\Phi^T(\eta, \mathbf{1}, \mathbf{1})$, i.e. when $R_1=R_2=\mathbf{1}$ in~\eqref{T},   is irreducibly covariant with respect to quaternion group $Q$.
\end{corollary}
We can say even more by exploiting the connection between $ICLM$ generated by the irreducible representations of $Q$ and $S(3)$, described briefly in Section~\ref{S_ICLM}. Firstly we present  what types of linear transformations on $\operatorname{End}[\mathbb{R}^{3}]$ are induced  by $ICLM$ generated by the groups $Q$ and $S(3)$. We have:
\begin{proposition}
	\label{prop2}
	\begin{enumerate}[a)]
		\item The $ICLM$  $\Phi ^{Q}(l)=\Pi ^{\id}+\sum_{i}l_{i}\Pi ^{i}$, where $l=(l_i)=(l_{t_1},l_{t_2},l_{t_3})$  induces
		the  following map in $\mathbb{M}(3,\mathbb{R})$  
		\be
		\Phi ^{Q}(l)[\rho (r)]=\rho (r^{\prime }),\qquad r^{\prime }=M^{Q}(r),\qquad
		M^{Q}(l)=\operatorname{diag}(l_{t_3},l_{t_1},l_{t_2}).
		\ee
		
		\item The $ICLM$ $\Phi ^{S(3)}(\widetilde{l})=\Pi ^{\id}+\widetilde{l}_{\sgn}\Pi ^{\sgn}+\widetilde{l}_{\lambda }\Pi
		^{\lambda }$ induces
		the  following map in $\mathbb{M}(3,\mathbb{R})$ 
		\be
		\Phi ^{S(3)}(\widetilde{l})[\rho (r)]=\rho (r'),\qquad r'=M^{S(3)}(r),\qquad M^{S(3)}(k)=\operatorname{diag}(\widetilde{l}_{\lambda },\widetilde{l}_{\lambda },\widetilde{l}_{\sgn}).
		\ee
	\end{enumerate}
\end{proposition}

Both  maps induce the diagonal transformation of the Bloch
vector but the map $\Phi ^{Q}(l)$ is more general than $\Phi ^{S(3)}(\widetilde{l})$. What is more, we have
\begin{corollary}
	\label{c18}
	The following relation between  the maps  $\Phi ^{Q}(l),$ $\Phi
	^{S(3)}(\widetilde{l})$ and the general form of the unital qubit quantum channel $T(\eta ,R_{1},R_{2}) 
	$ holds
	\be
	M^{\Phi }(l)=T(\eta ^{Q},\id_{SO(3)},\id_{SO(3)})\Rightarrow l_{1}=\eta
	_{2}^{Q},\quad l_{2}=\eta _{3}^{Q},\quad l_{3}=\eta _{1}^{Q},
	\ee
	\be
	M^{S(3)}(\widetilde{l})=T(\eta ^{S(3)},\id_{SO(3)},\id_{SO(3)})\Rightarrow \widetilde{l}_{\lambda
	}=\eta _{1}^{S(3)},\quad \widetilde{l}_{\lambda }=\eta _{2}^{S(3)},\quad \widetilde{l}_{\sgn}=\eta
	_{3}^{S(3)}.
	\ee
	The spectral parameters $l=(l_{i})$ and $\widetilde{l}=(\widetilde{l}_{\sgn},\widetilde{l}_{\lambda
	})$ of the   $ICLMs$ $\Phi ^{Q}(l)$ and $\Phi ^{S(3)}(\widetilde{l})$ are
	precisely the $SSV$ of the corresponding maps $T(\eta
	^{Q},\id_{SO(3)},\id_{SO(3)})$ and   $T(\eta ^{S(3)},\id_{SO(3)},\id_{SO(3)})$. 
\end{corollary}
Under the identification of the parameters $(l_{i})$  and $(\eta _{i}^{Q})$
in  Corollary~\ref{c18}, we may formulate the following

\begin{proposition}
	\label{prop19}
	Any qubit unital linear map can be decomposed as follows 
	\be
	T\equiv T(l,R_{1},R_{2})=R_{1}M^{Q}(l)R_{2},
	\ee
	where $M^{Q}(l)$ is the matrix induced by some quaternion irreducible covariant linear map $\Phi ^{Q}(l)$. In other words, the $SSV$ of any qubit unital map $T\in \operatorname{End}[\mathbb{R}^{3}]$ are spectral parameters of some quaternion irreducible covariant linear map $\Phi ^{Q}(l)$.
\end{proposition}

We know from~\cite{Kubler} that all important quantum properties of the qubit unital  maps depend only on their $SSV$. In fact, investigated properties depend on the spectral parameters of the quaternion $ICLM$ $\Phi ^{Q}(l)$. Namely,  the Fujiwara-Algoet conditions for $CP$ of the qubit unital map $T(\eta
,R_{1},R_{2})$~\cite{Kubler} 
\be
\label{Algoet}
1+\eta _{3}\geq |\eta _{1}+\eta _{2}|,\quad 1-\eta _{3}\geq |\eta _{1}-\eta
_{2}|
\ee
are exactly the same as our conditions on spectral parameters $l=(l_{i})$,
derived in~\cite{Moz2017} (or see~\eqref{CPQ} in this paper).
Moreover, the conditions for $P$  of the qubit unital  map $T(\eta ,R_{1},R_{2})$,
given in~\cite{Kubler} 
\be
|\eta _{i}|\leq 1,\quad i=1,2,3
\ee
are also the same as our conditions for $P$ of the irreducibly covariant quaternion map $\Phi (l)$ presented in Theorem~\ref{thmQQ}.
One can see that the Fujiwara-Algoet conditions~\cite{PhysRevA.59.3290} for $CP$ presented in~\eqref{Algoet} for the qubit unital map, 
in the case  of the map $T(\eta ^{S(3)},\id_{SO(3)},\id_{SO(3)})$, from Corollary~\ref{c18}, induced by  $\Phi ^{S(3)}(\widetilde{l})$, lead to the conditions for $CP$ of $\Phi
^{S(3)}(\widetilde{l})$
\be
1+\widetilde{l}_{\sgn}\geq 2|\widetilde{l}_{\lambda }|,\quad 1\geq |\widetilde{l}_{\sgn}|.
\ee
We can summarize the above as follows

\begin{proposition}
	The qubit unital linear map $T(l,R_{1},R_{2})$, given in Proposition~\ref{prop19}, is $CP$ (respectively $P$), if and only
	if the irreducibly covariant quaternion map $\Phi (l)$ is $CP$ (respectively $P$%
	).
\end{proposition}

All qubit unital linear maps  depend on the elements of the group $SO(3)$, and their specific properties ($CP, P$) are determined
by $ICLMs$ of the quaternion group $Q$, which are simpler object in the general description.
Moreover,  similar forms of the matrices $M^{Q}(l)=\operatorname{diag}(l_{t_3},l_{t_1},l_{t_2})$
and $M^{S(3)}(\widetilde{l})=\operatorname{diag}(\widetilde{l}_{\lambda },\widetilde{l}_{\lambda },\widetilde{l}_{\sgn})$, presented in Proposition~\ref{prop2}
suggest that in the particular case the  $ICLM$ $\Phi ^{Q}(l)$ of the form $%
\Phi ^{Q}(\widetilde{l}_{\lambda },\widetilde{l}_{\lambda },\widetilde{l}_{\sgn})$ is closely related to $ICLM$
$\Phi ^{S(3)}(\widetilde{l}_{\lambda },\widetilde{l}_{\sgn})$.
	\section{Conclusions and discussion}
    The results presented in the paper show how symmetric patterns imposed on linear maps via (irreducible) covariance make examination of its certain properties, like positivity simpler than in the general case. A variety of techniques was employed which led to different (both general and partial) results concerning multiple linear maps. In particular the spectral decomposition of an $ICLM$ allows to formulate the conditions for the positivity as the conditions on the mentioned spectrum of $ICLM$, leading to  geometrical interpretations in the spaces of $ICLM'$s spectra. Basing on this idea, the full description of positivity and complete positivity for low dimensional of $ICLMs$ induced by the permutation group $S(3)$ and the quaternion group $Q$ is given. In higher dimensions we present methods allowing us to find  the regions for positivity exploiting method inspired by the inverse reduction map, as well as we exploited the direct approach, working remarkable effective for monomial unitary group $MU(d)$. We show that the $ICLMs$ induced by the monomial unitary group reproduce the generalised Choi map, or conversely the Choi map exhibits covariant properties.
	Additionally, the connection of Fujiwara-Algolet conditions for an arbitrary qubit unital map with the quaternion symmetry is shown. 

	Still many questions remain open. For example, the $ICLM$ arising from irreducible representation of symmetric group $S(4)$ was examined via inverse reduction map and thus the obtained region of positivity is not maximal. The general solution despite attempts could not be obtained. Moreover, the (non-)decomposability of $ICLM$ with respect to unitary monomial subgroup was not resolved in the general and thus there may be some region from which new, non-decomposable entanglement witnesses may arise. Moreover, the results on Fujiwara-Algolet conditions suggest the possibility of similar  relation for unital maps and some, possibly different symmetries in higher dimensions.

	\section*{Acknowledgements}
	MS~was supported by the grant "Mobilno{\'s}{\'c} Plus IV", 1271/MOB/IV/2015/0 from the Polish Ministry of Science and Higher Education. MS, PK would like to thank DAMTP (University of Cambridge) where substantial part of this work has been done. PK's visit at the University of Cambridge was funded from programme PROM from the Polish National Agency for Academic Exchange.
\appendix

\section{{An explicit method of constructing matrix form of the irreducible representation of the symmetric group $S(n)$ for the partition $(n-1,1)$}}
\label{NovelSn}
We consider the matrix form of the natural representation of $S(n)$ given by
the permutation matrices 
\be
\label{cycle}
\forall \sigma \in S(n)\quad M(\sigma )=(\delta _{i\sigma (j)}),\quad
i,j=0,1,\ldots,n-1,\quad M(\sigma )\in \M(n,\mathbb{C}). 
\ee
The cycle $c=(0,1,\ldots,n-1)$ acts in the standard way $c(i)=i+_{n}1$, where
the addition is modulo $n.$ Let us formulate the following

\begin{definition}
	\label{U}
	Let $\varepsilon $ be a primitive $n$-th root of unity, we define the matrix $%
	U\in \M(n,\mathbb{C})$ as
	\be
	U=(u_{kl}):u_{kl}=\frac{1}{\sqrt{n}}\varepsilon ^{kl},\quad k,l=0,1,\ldots,n-1. 
	\ee
\end{definition}

It is easy to check that

\begin{proposition}
	The matrix $U$ from Definition~\ref{U} is unitary and diagonalizes the cycle matrix $M$ in~\eqref{cycle} for the cycle $c=(0,1,\ldots,n-1)$, i.e. we
	have 
	\be
	U^{\dagger}M(c)U=\operatorname{diag}(1,\overline{\varepsilon },\overline{\varepsilon }%
	^{2},\ldots,\overline{\varepsilon }^{n-1}). 
	\ee
\end{proposition}

It appears, that the matrix $U$ not only diagonalises the matrix $M(c)$, but
also reduces the natural, permutation matrix representation of $M(\sigma )$ in~\eqref{cycle}
of $S(n)$ into  the direct sum of irreducible representation. In fact we have

\begin{proposition}
	For any permutation $\sigma \in S(n)$ a matrix 
	\be
	\Psi (\sigma )=U^{\dagger}M(\sigma )U=(\Psi _{ij}(\sigma )):\Psi _{ij}(\sigma )=%
	\frac{1}{n}\sum_{l=0}^{n-1}\varepsilon ^{jl-\sigma (l)i},\quad
	i,j=0,1,\ldots,n-1 
	\ee
	has the following property 
	\be
	\forall \sigma \in S(n)\quad \Psi _{00}(\sigma )=1,\quad \Psi _{0j}(\sigma
	)=0=\Psi _{i0}(\sigma ),\quad i,j=1,\ldots,n-1, 
	\ee
	which produces the block form of $\Psi(\sigma)$
	\be
	\Psi (\sigma )=\left( 
	\begin{array}{cc}
		1 & 0 \\ 
		0 & \psi (\sigma )%
	\end{array}%
	\right) :\psi (\sigma )\in \M(n-1,\mathbb{C}). 
	\ee
\end{proposition}

The above proposition allows us the formulate the following

\begin{theorem}
	\label{cxz}
	The map $\psi :S(n)\rightarrow \M(n-1,\mathbb{C})$ of the form 
	\be
	\forall \sigma \in S(n)\quad \psi (\sigma )=(\psi _{ij}(\sigma
	)):\psi _{ij}(\sigma )=\frac{1}{n}\sum_{l=0}^{n-1}\varepsilon ^{jl-\sigma
		(l)i},\quad i,j=1,\ldots,n-1 
	\ee
	is the irreducible representation of $S(n)$ corresponding to the partition $
	(n-1,1)$.
\end{theorem}

After calculations we get the following

\begin{proposition}
	\label{MatrixIrrep}
	The cycle $c=(0,1,\ldots,n-1)$ in irreducible representation given through Theorem~\ref{cxz}  has the form 
	\be
	\psi (c)=\operatorname{diag}(\overline{\varepsilon },\overline{\varepsilon }^{2},\ldots,
	\overline{\varepsilon }^{n-1}),
	\ee
	whereas the transposition $\sigma =(ab)$ is represented by a matrix 
	\be
	\psi (ab)=(\psi _{ij}(ab)):\psi _{ij}(ab)=\delta _{ij}+\frac{1}{n}%
	(\varepsilon ^{-ai}-\varepsilon ^{-bi})(\varepsilon ^{bj}-\varepsilon ^{aj}).
	\ee
\end{proposition}

Next, using Proposition~\ref{MatrixIrrep} we calculate two examples presenting the matrix forms of the irreducible representations of the symmetric groups $S(3),S(4)$ and the partitions $(2,1), (3,1)$ respectively.

\begin{example}
	The irreducible representation labelled by $(2,1)$ of the group $S(3)$, has the following $\varepsilon $-form%
	\be
	\psi (01)=\begin{pmatrix}
		0 & \varepsilon ^{2} \\ 
		\varepsilon  & 0%
	\end{pmatrix} ,\quad \psi (02)=\begin{pmatrix}
		0 & \varepsilon  \\ 
		\varepsilon ^{2} & 0%
	\end{pmatrix},\quad \psi (12)=\begin{pmatrix}
		0 & 1 \\ 
		1 & 0%
	\end{pmatrix} ,\quad 
	\psi (012)=\begin{pmatrix}
		\varepsilon ^{2} & 0 \\ 
		0 & \varepsilon 
	\end{pmatrix} ,\quad \psi (021)=\begin{pmatrix}
		\varepsilon  & 0 \\ 
		0 & \varepsilon ^{2}%
	\end{pmatrix}.
	\ee
\end{example}
\begin{example}
	The irreducible representation labelled by $(3,1)$ of the group $S(4)$ has the following $\epsilon$-form, listed in Table \ref{t2} below.
	
\begin{table}[ht!]	
	\begin{tabular}{|l|l|}
		\hline
		$\psi(23)=\left(
		\begin{array}{ccc}
		\frac{1}{2} & \frac{1}{2}+\frac{i}{2} & -\frac{i}{2} \\
		\frac{1}{2}-\frac{i}{2} & 0 & \frac{1}{2}+\frac{i}{2} \\
		\frac{i}{2} & \frac{1}{2}-\frac{i}{2} & \frac{1}{2} \\
		\end{array}
		\right)$ & $\psi(12)=\left(
		\begin{array}{ccc}
		\frac{1}{2} & \frac{1}{2}-\frac{i}{2} & \frac{i}{2} \\
		\frac{1}{2}+\frac{i}{2} & 0 & \frac{1}{2}-\frac{i}{2} \\
		-\frac{i}{2} & \frac{1}{2}+\frac{i}{2} & \frac{1}{2} \\
		\end{array}
		\right)$ \\[3ex]
		\hline
		$\psi((23)(13))=\left(
		\begin{array}{ccc}
		-\frac{i}{2} & \frac{1}{2}+\frac{i}{2} & \frac{1}{2} \\
		\frac{1}{2}+\frac{i}{2} & 0 & \frac{1}{2}-\frac{i}{2} \\
		\frac{1}{2} & \frac{1}{2}-\frac{i}{2} & \frac{i}{2} \\
		\end{array}
		\right)$ & $\psi((23)(12))=\left(
		\begin{array}{ccc}
		\frac{i}{2} & \frac{1}{2}-\frac{i}{2} & \frac{1}{2} \\
		\frac{1}{2}-\frac{i}{2} & 0 & \frac{1}{2}+\frac{i}{2} \\
		\frac{1}{2} & \frac{1}{2}+\frac{i}{2} & -\frac{i}{2} \\
		\end{array}
		\right)$\\[3ex]
		\hline
		$\psi(13)=\left(
		\begin{array}{ccc}
		0 & 0 & 1 \\
		0 & 1 & 0 \\
		1 & 0 & 0 \\
		\end{array}
		\right)$ & $\psi(01)=\left(
		\begin{array}{ccc}
		\frac{1}{2} & -\frac{1}{2}-\frac{i}{2} & -\frac{i}{2} \\
		-\frac{1}{2}+\frac{i}{2} & 0 & -\frac{1}{2}-\frac{i}{2} \\
		\frac{i}{2} & -\frac{1}{2}+\frac{i}{2} & \frac{1}{2} \\
		\end{array}
		\right)$\\[3ex]
		\hline
		$\psi((01)(23))=\left(
		\begin{array}{ccc}
		0 & 0 & -i \\
		0 & -1 & 0 \\
		i & 0 & 0 \\
		\end{array}
		\right)$ & $\psi((12)(02))=\left(
		\begin{array}{ccc}
		-\frac{i}{2} & \frac{1}{2}-\frac{i}{2} & -\frac{1}{2} \\
		-\frac{1}{2}+\frac{i}{2} & 0 & -\frac{1}{2}-\frac{i}{2} \\
		-\frac{1}{2} & \frac{1}{2}+\frac{i}{2} & \frac{i}{2} \\
		\end{array}
		\right)$\\[3ex]
		\hline
		$\psi((12)(23)(03))=\left(
		\begin{array}{ccc}
		-i & 0 & 0 \\
		0 & -1 & 0 \\
		0 & 0 & i \\
		\end{array}
		\right)$ & $\psi((13)(23)(02))=\left(
		\begin{array}{ccc}
		-\frac{1}{2} & \frac{1}{2}-\frac{i}{2} & -\frac{i}{2} \\
		-\frac{1}{2}-\frac{i}{2} & 0 & -\frac{1}{2}+\frac{i}{2} \\
		\frac{i}{2} & \frac{1}{2}+\frac{i}{2} & -\frac{1}{2} \\
		\end{array}
		\right)$\\[3ex]
		\hline
		$\psi((13)(03))=\left(
		\begin{array}{ccc}
		-\frac{i}{2} & -\frac{1}{2}-\frac{i}{2} & \frac{1}{2} \\
		-\frac{1}{2}-\frac{i}{2} & 0 & -\frac{1}{2}+\frac{i}{2} \\
		\frac{1}{2} & -\frac{1}{2}+\frac{i}{2} & \frac{i}{2} \\
		\end{array}
		\right)$ & $\psi((12)(01))=\left(
		\begin{array}{ccc}
		\frac{i}{2} & -\frac{1}{2}-\frac{i}{2} & -\frac{1}{2} \\
		\frac{1}{2}+\frac{i}{2} & 0 & \frac{1}{2}-\frac{i}{2} \\
		-\frac{1}{2} & -\frac{1}{2}+\frac{i}{2} & -\frac{i}{2} \\
		\end{array}
		\right)$\\[3ex]
		\hline
		$\psi((23)(13)(01))=\left(
		\begin{array}{ccc}
		-\frac{1}{2} & -\frac{1}{2}+\frac{i}{2} & -\frac{i}{2} \\
		\frac{1}{2}+\frac{i}{2} & 0 & \frac{1}{2}-\frac{i}{2} \\
		\frac{i}{2} & -\frac{1}{2}-\frac{i}{2} & -\frac{1}{2} \\
		\end{array}
		\right)$ & $\psi(02)=\left(
		\begin{array}{ccc}
		0 & 0 & -1 \\
		0 & 1 & 0 \\
		-1 & 0 & 0 \\
		\end{array}
		\right)$\\[3ex]
		\hline
		$\psi((23)(03))=\left(
		\begin{array}{ccc}
		-\frac{i}{2} & -\frac{1}{2}+\frac{i}{2} & -\frac{1}{2} \\
		\frac{1}{2}-\frac{i}{2} & 0 & \frac{1}{2}+\frac{i}{2} \\
		-\frac{1}{2} & -\frac{1}{2}-\frac{i}{2} & \frac{i}{2} \\
		\end{array}
		\right)$ & $\psi((02)(13))=\left(
		\begin{array}{ccc}
		-1 & 0 & 0 \\
		0 & 1 & 0 \\
		0 & 0 & -1 \\
		\end{array}
		\right)$\\[3ex]
		\hline
		$\psi((12)(13)(03))=\left(
		\begin{array}{ccc}
		-\frac{1}{2} & -\frac{1}{2}-\frac{i}{2} & \frac{i}{2} \\
		\frac{1}{2}-\frac{i}{2} & 0 & \frac{1}{2}+\frac{i}{2} \\
		-\frac{i}{2} & -\frac{1}{2}+\frac{i}{2} & -\frac{1}{2} \\
		\end{array}
		\right)$ & $\psi((23)(12)(01))=\left(
		\begin{array}{ccc}
		i & 0 & 0 \\
		0 & -1 & 0 \\
		0 & 0 & -i \\
		\end{array}
		\right)$\\[3ex]
		\hline
		$\psi((13)(01))=\left(
		\begin{array}{ccc}
		\frac{i}{2} & -\frac{1}{2}+\frac{i}{2} & \frac{1}{2} \\
		-\frac{1}{2}+\frac{i}{2} & 0 & -\frac{1}{2}-\frac{i}{2} \\
		\frac{1}{2} & -\frac{1}{2}-\frac{i}{2} & -\frac{i}{2} \\
		\end{array}
		\right)$ & $\psi((23)(02))=\left(
		\begin{array}{ccc}
		\frac{i}{2} & \frac{1}{2}+\frac{i}{2} & -\frac{1}{2} \\
		-\frac{1}{2}-\frac{i}{2} & 0 & -\frac{1}{2}+\frac{i}{2} \\
		-\frac{1}{2} & \frac{1}{2}-\frac{i}{2} & -\frac{i}{2} \\
		\end{array}
		\right)$ \\[3ex]
		\hline
		$\psi(03)=\left(
		\begin{array}{ccc}
		\frac{1}{2} & -\frac{1}{2}+\frac{i}{2} & \frac{i}{2} \\
		-\frac{1}{2}-\frac{i}{2} & 0 & -\frac{1}{2}+\frac{i}{2} \\
		-\frac{i}{2} & -\frac{1}{2}-\frac{i}{2} & \frac{1}{2} \\
		\end{array}
		\right)$ & $\psi((13)(12)(02))=\left(
		\begin{array}{ccc}
		-\frac{1}{2} & \frac{1}{2}+\frac{i}{2} & \frac{i}{2} \\
		-\frac{1}{2}+\frac{i}{2} & 0 & -\frac{1}{2}-\frac{i}{2} \\
		-\frac{i}{2} & \frac{1}{2}-\frac{i}{2} & -\frac{1}{2} \\
		\end{array}
		\right)$\\[3ex]
		\hline
		$\psi((03)(12))=\left(
		\begin{array}{ccc}
		0 & 0 & i \\
		0 & -1 & 0 \\
		-i & 0 & 0 \\
		\end{array}
		\right)$ & \text{\phantom{o}}\\[3ex]
		\hline
	\end{tabular}
\caption{Table presents the $\epsilon-$matrix representations of the irreducible representation labelled by the partition $(3,1)$ for the permutation group $S(4)$. }
	\label{t2}
\end{table}
\end{example}	

\FloatBarrier
	\bibliographystyle{apsrev}
	\bibliography{um}

\begin{thebibliography}{33}
\expandafter\ifx\csname natexlab\endcsname\relax\def\natexlab#1{#1}\fi
\expandafter\ifx\csname bibnamefont\endcsname\relax
  \def\bibnamefont#1{#1}\fi
\expandafter\ifx\csname bibfnamefont\endcsname\relax
  \def\bibfnamefont#1{#1}\fi
\expandafter\ifx\csname citenamefont\endcsname\relax
  \def\citenamefont#1{#1}\fi
\expandafter\ifx\csname url\endcsname\relax
  \def\url#1{\texttt{#1}}\fi
\expandafter\ifx\csname urlprefix\endcsname\relax\def\urlprefix{URL }\fi
\providecommand{\bibinfo}[2]{#2}
\providecommand{\eprint}[2][]{\url{#2}}

\bibitem[{\citenamefont{Gyongyosi et~al.}(2018)\citenamefont{Gyongyosi, Imre,
  and Nguyen}}]{rev1}
\bibinfo{author}{\bibfnamefont{L.}~\bibnamefont{Gyongyosi}},
  \bibinfo{author}{\bibfnamefont{S.}~\bibnamefont{Imre}}, \bibnamefont{and}
  \bibinfo{author}{\bibfnamefont{H.~V.} \bibnamefont{Nguyen}},
  \bibinfo{journal}{IEEE Communications Surveys and Tutorials}
  \textbf{\bibinfo{volume}{20}}, \bibinfo{pages}{1149–1205}
  (\bibinfo{year}{2018}), ISSN \bibinfo{issn}{2373-745X},
  \urlprefix\url{http://dx.doi.org/10.1109/COMST.2017.2786748}.

\bibitem[{\citenamefont{Terhal}(2000)}]{ter}
\bibinfo{author}{\bibfnamefont{B.~M.} \bibnamefont{Terhal}},
  \bibinfo{journal}{Physics Letters A} \textbf{\bibinfo{volume}{271}},
  \bibinfo{pages}{319–326} (\bibinfo{year}{2000}), ISSN
  \bibinfo{issn}{0375-9601},
  \urlprefix\url{http://dx.doi.org/10.1016/S0375-9601(00)00401-1}.

\bibitem[{\citenamefont{Woronowicz}(1976)}]{Woronowicz}
\bibinfo{author}{\bibfnamefont{S.}~\bibnamefont{Woronowicz}},
  \bibinfo{journal}{Reports on Mathematical Physics}
  \textbf{\bibinfo{volume}{10}}, \bibinfo{pages}{165 } (\bibinfo{year}{1976}),
  ISSN \bibinfo{issn}{0034-4877},
  \urlprefix\url{http://www.sciencedirect.com/science/article/pii/0034487776900380}.

\bibitem[{\citenamefont{Chruściński and Sarbicki}(2014)}]{Chrust1}
\bibinfo{author}{\bibfnamefont{D.}~\bibnamefont{Chruściński}}
  \bibnamefont{and} \bibinfo{author}{\bibfnamefont{G.}~\bibnamefont{Sarbicki}},
  \bibinfo{journal}{Journal of Physics A: Mathematical and Theoretical}
  \textbf{\bibinfo{volume}{47}}, \bibinfo{pages}{483001}
  (\bibinfo{year}{2014}), ISSN \bibinfo{issn}{1751-8121},
  \urlprefix\url{http://dx.doi.org/10.1088/1751-8113/47/48/483001}.

\bibitem[{\citenamefont{Takasaki and Tomiyama}(1983)}]{Takasaki1983}
\bibinfo{author}{\bibfnamefont{T.}~\bibnamefont{Takasaki}} \bibnamefont{and}
  \bibinfo{author}{\bibfnamefont{J.}~\bibnamefont{Tomiyama}},
  \bibinfo{journal}{Mathematische Zeitschrift} \textbf{\bibinfo{volume}{184}},
  \bibinfo{pages}{101} (\bibinfo{year}{1983}), ISSN \bibinfo{issn}{1432-1823},
  \urlprefix\url{https://doi.org/10.1007/BF01162009}.

\bibitem[{\citenamefont{Tomiyama}(1985)}]{TOMIYAMA1985169}
\bibinfo{author}{\bibfnamefont{J.}~\bibnamefont{Tomiyama}},
  \bibinfo{journal}{Linear Algebra and its Applications}
  \textbf{\bibinfo{volume}{69}}, \bibinfo{pages}{169 } (\bibinfo{year}{1985}),
  ISSN \bibinfo{issn}{0024-3795},
  \urlprefix\url{http://www.sciencedirect.com/science/article/pii/0024379585900746}.

\bibitem[{\citenamefont{Chruściński and
  Kossakowski}(2009{\natexlab{a}})}]{Chru2009}
\bibinfo{author}{\bibfnamefont{D.}~\bibnamefont{Chruściński}}
  \bibnamefont{and}
  \bibinfo{author}{\bibfnamefont{A.}~\bibnamefont{Kossakowski}},
  \bibinfo{journal}{Communications in Mathematical Physics}
  \textbf{\bibinfo{volume}{290}}, \bibinfo{pages}{1051–1064}
  (\bibinfo{year}{2009}{\natexlab{a}}), ISSN \bibinfo{issn}{1432-0916},
  \urlprefix\url{http://dx.doi.org/10.1007/s00220-009-0790-8}.

\bibitem[{\citenamefont{Kossakowski}(2003)}]{Kossakowski2003}
\bibinfo{author}{\bibfnamefont{A.}~\bibnamefont{Kossakowski}},
  \bibinfo{journal}{Open Systems \& Information Dynamics}
  \textbf{\bibinfo{volume}{10}}, \bibinfo{pages}{213} (\bibinfo{year}{2003}),
  ISSN \bibinfo{issn}{1573-1324},
  \urlprefix\url{https://doi.org/10.1023/A:1025101606680}.

\bibitem[{\citenamefont{Chruściński}(2014)}]{Chru2014}
\bibinfo{author}{\bibfnamefont{D.}~\bibnamefont{Chruściński}},
  \bibinfo{journal}{Open Systems \& Information Dynamics}
  \textbf{\bibinfo{volume}{21}}, \bibinfo{pages}{1450001}
  (\bibinfo{year}{2014}), ISSN \bibinfo{issn}{1793-7191},
  \urlprefix\url{http://dx.doi.org/10.1142/S1230161214500012}.

\bibitem[{\citenamefont{Chruściński and
  Kossakowski}(2009{\natexlab{b}})}]{Chru2009a}
\bibinfo{author}{\bibfnamefont{D.}~\bibnamefont{Chruściński}}
  \bibnamefont{and}
  \bibinfo{author}{\bibfnamefont{A.}~\bibnamefont{Kossakowski}},
  \bibinfo{journal}{Physics Letters A} \textbf{\bibinfo{volume}{373}},
  \bibinfo{pages}{2301–2305} (\bibinfo{year}{2009}{\natexlab{b}}), ISSN
  \bibinfo{issn}{0375-9601},
  \urlprefix\url{http://dx.doi.org/10.1016/j.physleta.2009.04.068}.

\bibitem[{\citenamefont{Tang}(1986)}]{TANG198633}
\bibinfo{author}{\bibfnamefont{W.-S.} \bibnamefont{Tang}},
  \bibinfo{journal}{Linear Algebra and its Applications}
  \textbf{\bibinfo{volume}{79}}, \bibinfo{pages}{33 } (\bibinfo{year}{1986}),
  ISSN \bibinfo{issn}{0024-3795},
  \urlprefix\url{http://www.sciencedirect.com/science/article/pii/0024379586902909}.

\bibitem[{\citenamefont{Bhat}(2011)}]{bhat2011}
\bibinfo{author}{\bibfnamefont{B.~V.~R.} \bibnamefont{Bhat}},
  \bibinfo{journal}{Banach J. Math. Anal.} \textbf{\bibinfo{volume}{5}},
  \bibinfo{pages}{1} (\bibinfo{year}{2011}),
  \urlprefix\url{https://doi.org/10.15352/bjma/1313362996}.

\bibitem[{\citenamefont{Collins et~al.}(2018)\citenamefont{Collins, Osaka, and
  Sapra}}]{COLLINS2018398}
\bibinfo{author}{\bibfnamefont{B.}~\bibnamefont{Collins}},
  \bibinfo{author}{\bibfnamefont{H.}~\bibnamefont{Osaka}}, \bibnamefont{and}
  \bibinfo{author}{\bibfnamefont{G.}~\bibnamefont{Sapra}},
  \bibinfo{journal}{Linear Algebra and its Applications}
  \textbf{\bibinfo{volume}{555}}, \bibinfo{pages}{398 } (\bibinfo{year}{2018}),
  ISSN \bibinfo{issn}{0024-3795},
  \urlprefix\url{http://www.sciencedirect.com/science/article/pii/S0024379518302957}.

\bibitem[{\citenamefont{Bardet et~al.}(2018)\citenamefont{Bardet, Collins, and
  Sapra}}]{bardet2018characterization}
\bibinfo{author}{\bibfnamefont{I.}~\bibnamefont{Bardet}},
  \bibinfo{author}{\bibfnamefont{B.}~\bibnamefont{Collins}}, \bibnamefont{and}
  \bibinfo{author}{\bibfnamefont{G.}~\bibnamefont{Sapra}},
  \emph{\bibinfo{title}{Characterization of equivariant maps and application to
  entanglement detection}} (\bibinfo{year}{2018}), \eprint{1811.08193}.

\bibitem[{\citenamefont{Huber}(2020)}]{huber2020positive}
\bibinfo{author}{\bibfnamefont{F.}~\bibnamefont{Huber}},
  \emph{\bibinfo{title}{Positive maps and matrix contractions from the
  symmetric group}} (\bibinfo{year}{2020}), \eprint{2002.12887}.

\bibitem[{\citenamefont{Siudzińska and Chruściński}(2018)}]{Siudzi_ska_2018}
\bibinfo{author}{\bibfnamefont{K.}~\bibnamefont{Siudzińska}} \bibnamefont{and}
  \bibinfo{author}{\bibfnamefont{D.}~\bibnamefont{Chruściński}},
  \bibinfo{journal}{Journal of Mathematical Physics}
  \textbf{\bibinfo{volume}{59}}, \bibinfo{pages}{033508}
  (\bibinfo{year}{2018}), ISSN \bibinfo{issn}{1089-7658},
  \urlprefix\url{http://dx.doi.org/10.1063/1.5013604}.

\bibitem[{\citenamefont{Mozrzymas et~al.}(2017)\citenamefont{Mozrzymas,
  Studziński, and Datta}}]{Moz2017}
\bibinfo{author}{\bibfnamefont{M.}~\bibnamefont{Mozrzymas}},
  \bibinfo{author}{\bibfnamefont{M.}~\bibnamefont{Studziński}},
  \bibnamefont{and} \bibinfo{author}{\bibfnamefont{N.}~\bibnamefont{Datta}},
  \bibinfo{journal}{Journal of Mathematical Physics}
  \textbf{\bibinfo{volume}{58}}, \bibinfo{pages}{052204}
  (\bibinfo{year}{2017}), ISSN \bibinfo{issn}{1089-7658},
  \urlprefix\url{http://dx.doi.org/10.1063/1.4983710}.

\bibitem[{\citenamefont{Mozrzymas et~al.}(2015)\citenamefont{Mozrzymas,
  Rutkowski, and Studziński}}]{Moz2015}
\bibinfo{author}{\bibfnamefont{M.}~\bibnamefont{Mozrzymas}},
  \bibinfo{author}{\bibfnamefont{A.}~\bibnamefont{Rutkowski}},
  \bibnamefont{and}
  \bibinfo{author}{\bibfnamefont{M.}~\bibnamefont{Studziński}},
  \bibinfo{journal}{Journal of Physics A: Mathematical and Theoretical}
  \textbf{\bibinfo{volume}{48}}, \bibinfo{pages}{395302}
  (\bibinfo{year}{2015}), ISSN \bibinfo{issn}{1751-8121},
  \urlprefix\url{http://dx.doi.org/10.1088/1751-8113/48/39/395302}.

\bibitem[{\citenamefont{Fran{\c{c}}a and Hashagen}(2018)}]{Daniel}
\bibinfo{author}{\bibfnamefont{D.~S.} \bibnamefont{Fran{\c{c}}a}}
  \bibnamefont{and} \bibinfo{author}{\bibfnamefont{A.~K.}
  \bibnamefont{Hashagen}}, \bibinfo{journal}{Journal of Physics A: Mathematical
  and Theoretical} \textbf{\bibinfo{volume}{51}}, \bibinfo{pages}{395302}
  (\bibinfo{year}{2018}),
  \urlprefix\url{https://doi.org/10.1088%2F1751-8121%2Faad6fa}.

\bibitem[{\citenamefont{den Nest}(2011)}]{Nest}
\bibinfo{author}{\bibfnamefont{M.~V.} \bibnamefont{den Nest}},
  \bibinfo{journal}{New Journal of Physics} \textbf{\bibinfo{volume}{13}},
  \bibinfo{pages}{123004} (\bibinfo{year}{2011}),
  \urlprefix\url{https://doi.org/10.1088%2F1367-2630%2F13%2F12%2F123004}.

\bibitem[{\citenamefont{Nielsen and Chuang}(2010)}]{NielsenChuang}
\bibinfo{author}{\bibfnamefont{M.~A.} \bibnamefont{Nielsen}} \bibnamefont{and}
  \bibinfo{author}{\bibfnamefont{I.~L.} \bibnamefont{Chuang}},
  \emph{\bibinfo{title}{Quantum Computation and Quantum Information}}
  (\bibinfo{publisher}{Cambridge University Press}, \bibinfo{year}{2010}).

\bibitem[{\citenamefont{Kye and Lee}(1992)}]{cho}
\bibinfo{author}{\bibfnamefont{S.~C.~S.} \bibnamefont{Kye}} \bibnamefont{and}
  \bibinfo{author}{\bibfnamefont{S.}~\bibnamefont{Lee}}, \bibinfo{journal}{Lin.
  Alg. Appl.} \textbf{\bibinfo{volume}{171}} (\bibinfo{year}{1992}).

\bibitem[{\citenamefont{Fujiwara and Algoet}(1999)}]{PhysRevA.59.3290}
\bibinfo{author}{\bibfnamefont{A.}~\bibnamefont{Fujiwara}} \bibnamefont{and}
  \bibinfo{author}{\bibfnamefont{P.}~\bibnamefont{Algoet}},
  \bibinfo{journal}{Phys. Rev. A} \textbf{\bibinfo{volume}{59}},
  \bibinfo{pages}{3290} (\bibinfo{year}{1999}),
  \urlprefix\url{https://link.aps.org/doi/10.1103/PhysRevA.59.3290}.

\bibitem[{\citenamefont{Jamiołkowski}(1972)}]{Jam}
\bibinfo{author}{\bibfnamefont{A.}~\bibnamefont{Jamiołkowski}},
  \bibinfo{journal}{Reports on Mathematical Physics}
  \textbf{\bibinfo{volume}{3}}, \bibinfo{pages}{275 } (\bibinfo{year}{1972}),
  ISSN \bibinfo{issn}{0034-4877},
  \urlprefix\url{http://www.sciencedirect.com/science/article/pii/0034487772900110}.

\bibitem[{\citenamefont{Fulton and Harris}(2004)}]{FHa}
\bibinfo{author}{\bibfnamefont{W.}~\bibnamefont{Fulton}} \bibnamefont{and}
  \bibinfo{author}{\bibfnamefont{J.}~\bibnamefont{Harris}},
  \emph{\bibinfo{title}{Representation Theory: A First Course}}
  (\bibinfo{publisher}{Springer}, \bibinfo{year}{2004}).

\bibitem[{\citenamefont{andA.I. Stern}(1982)}]{NS}
\bibinfo{author}{\bibfnamefont{M.~N.} \bibnamefont{andA.I. Stern}},
  \emph{\bibinfo{title}{Theory of Group Representations}}
  (\bibinfo{publisher}{Springer-Verlag New York}, \bibinfo{year}{1982}).

\bibitem[{\citenamefont{Størmer}(1982)}]{Stromer}
\bibinfo{author}{\bibfnamefont{E.}~\bibnamefont{Størmer}},
  \bibinfo{journal}{Proc. Amer. Math.Soc.} \textbf{\bibinfo{volume}{86}},
  \bibinfo{pages}{402} (\bibinfo{year}{1982}).

\bibitem[{\citenamefont{{Horodecki} and {Horodecki}}(1999)}]{Hor2}
\bibinfo{author}{\bibfnamefont{M.}~\bibnamefont{{Horodecki}}} \bibnamefont{and}
  \bibinfo{author}{\bibfnamefont{P.}~\bibnamefont{{Horodecki}}},
  \bibinfo{journal}{PRA} \textbf{\bibinfo{volume}{59}}, \bibinfo{pages}{4206}
  (\bibinfo{year}{1999}).

\bibitem[{\citenamefont{Sarbicki et~al.}(2016)\citenamefont{Sarbicki,
  Chru{\'{s}}ci{\'{n}}ski, and Mozrzymas}}]{ChrusWig}
\bibinfo{author}{\bibfnamefont{G.}~\bibnamefont{Sarbicki}},
  \bibinfo{author}{\bibfnamefont{D.}~\bibnamefont{Chru{\'{s}}ci{\'{n}}ski}},
  \bibnamefont{and}
  \bibinfo{author}{\bibfnamefont{M.}~\bibnamefont{Mozrzymas}},
  \bibinfo{journal}{Journal of Physics A: Mathematical and Theoretical}
  \textbf{\bibinfo{volume}{49}}, \bibinfo{pages}{305302}
  (\bibinfo{year}{2016}),
  \urlprefix\url{https://doi.org/10.1088%2F1751-8113%2F49%2F30%2F305302}.

\bibitem[{\citenamefont{Choi}(1975)}]{choi2}
\bibinfo{author}{\bibfnamefont{M.~D.} \bibnamefont{Choi}},
  \bibinfo{journal}{Linear Algebra and Appl.} \textbf{\bibinfo{volume}{12}}
  (\bibinfo{year}{1975}).

\bibitem[{\citenamefont{{Chruściński} and {Kossakowski}}(2008)}]{chru_kos}
\bibinfo{author}{\bibfnamefont{D.}~\bibnamefont{{Chruściński}}}
  \bibnamefont{and}
  \bibinfo{author}{\bibfnamefont{A.}~\bibnamefont{{Kossakowski}}},
  \bibinfo{journal}{Open Sys. and Information Dyn.}
  \textbf{\bibinfo{volume}{14}} (\bibinfo{year}{2008}).

\bibitem[{\citenamefont{{Chruściński}
  et~al.}(2008)\citenamefont{{Chruściński}, {Marciniak}, and
  {Rutkowski}}}]{chru_kos2}
\bibinfo{author}{\bibfnamefont{D.}~\bibnamefont{{Chruściński}}},
  \bibinfo{author}{\bibfnamefont{M.}~\bibnamefont{{Marciniak}}},
  \bibnamefont{and}
  \bibinfo{author}{\bibfnamefont{A.}~\bibnamefont{{Rutkowski}}},
  \bibinfo{journal}{Open Sys. and Information Dyn.}
  \textbf{\bibinfo{volume}{14}} (\bibinfo{year}{2008}).

\bibitem[{\citenamefont{Kübler and Braun}(2018)}]{Kubler}
\bibinfo{author}{\bibfnamefont{J.~M.} \bibnamefont{Kübler}} \bibnamefont{and}
  \bibinfo{author}{\bibfnamefont{D.}~\bibnamefont{Braun}},
  \bibinfo{journal}{New Journal of Physics} \textbf{\bibinfo{volume}{20}},
  \bibinfo{pages}{083015} (\bibinfo{year}{2018}),
  \urlprefix\url{http://stacks.iop.org/1367-2630/20/i=8/a=083015}.

\end{thebibliography}
\end{document}